\documentclass[sigplan,screen]{acmart}

\usepackage{fontawesome}
\usepackage{amsmath,amsfonts,amscd}

\usepackage{amssymb}
\usepackage{fancyhdr}
\usepackage{mathrsfs}
\usepackage{graphicx}
\usepackage{textcomp}
\usepackage{subcaption}
\usepackage{multirow}
\usepackage{amsthm}
\usepackage[T1]{fontenc}
\usepackage{enumitem}
\usepackage{algorithm,algorithmicx,algpseudocode}
\usepackage{listings}
\usepackage{pifont}
\usepackage{bm}
\usepackage{booktabs}
\usepackage{threeparttable}

\newtheorem{assumption}{Assumption}

\DeclareMathOperator{\E}{\mathbb{E}}

\AtBeginDocument{%
  \providecommand\BibTeX{{%
    \normalfont B\kern-0.5em{\scshape i\kern-0.25em b}\kern-0.8em\TeX}}}

\acmYear{2022}
\acmConference[PPoPP '22]{27th ACM SIGPLAN Symposium on Principles and Practice of Parallel Programming}{April 2--6, 2022}{Seoul, Republic of Korea}
\acmDOI{10.1145/3503221.3508399}
\acmISBN{978-1-4503-9204-4/22/04}
\setcopyright{none}

\begin{document}

\title{Near-Optimal Sparse Allreduce for Distributed Deep Learning}


\author{Shigang Li}
\email{shigangli.cs@gmail.com}
\affiliation{%
  \institution{Department of Computer Science, ETH Zurich}
  \country{Switzerland}
}

\author{Torsten Hoefler}
\email{htor@inf.ethz.ch}
\affiliation{%
  \institution{Department of Computer Science, ETH Zurich}
  \country{Switzerland}
}

\renewcommand{\shortauthors}{Shigang Li and Torsten Hoefler}

\begin{abstract}

Communication overhead is one of the major obstacles to train large deep learning models at scale. Gradient sparsification is a promising technique to reduce the communication volume. However, it is very challenging to obtain real performance improvement because of  (1) the difficulty of achieving an scalable and efficient sparse \textit{allreduce} algorithm and (2) the sparsification overhead. This paper proposes O$k$-Top$k$, a scheme for distributed training with sparse gradients. O$k$-Top$k$ integrates a novel sparse allreduce algorithm (less than 6$k$ communication volume which is asymptotically optimal) with the decentralized parallel Stochastic Gradient Descent (SGD) optimizer, and its convergence is proved. To reduce the sparsification overhead, O$k$-Top$k$ efficiently selects the top-$k$ gradient values according to an estimated threshold. Evaluations are conducted on the Piz Daint supercomputer with neural network models from different deep learning domains. Empirical results show that O$k$-Top$k$ achieves similar model accuracy to dense allreduce. Compared with the optimized dense and the state-of-the-art sparse allreduces, O$k$-Top$k$ is more scalable and significantly improves training throughput (e.g., 3.29x-12.95x improvement for BERT on 256 GPUs).

\end{abstract}

\begin{CCSXML}
  <ccs2012>
  <concept>
  <concept_id>10003752.10003809.10010170</concept_id>
  <concept_desc>Theory of computation~Parallel algorithms</concept_desc>
  <concept_significance>500</concept_significance>
  </concept>
  <concept>
  <concept_id>10010147.10010257.10010293.10010294</concept_id>
  <concept_desc>Computing methodologies~Neural networks</concept_desc>
  <concept_significance>300</concept_significance>
  </concept>
  </ccs2012>
\end{CCSXML}

\ccsdesc[500]{Theory of computation~Parallel algorithms}
\ccsdesc[300]{Computing methodologies~Neural networks}

\keywords{distributed deep learning, allreduce, gradient sparsification, data parallelism}


\maketitle

\section{Introduction}

Training deep learning models is a major workload on large-scale computing systems. 
While such training may be parallelized in many ways~\cite{ben2019demystifying}, the dominant and simplest form is data parallelism.
In data-parallel training, the model is replicated across different compute nodes. After the computation of local gradient on each process is finished, the distributed gradients are accumulated across all processes, usually using an \textit{allreduce}~\cite{chan2007collective} operation.
However, not all dimensions are equally important and the communication of the distributed gradients can be sparsified significantly, introducing up to 99.9\% zero values without significant loss of accuracy. Only the nonzero values of the distributed gradients are accumulated across all processes.
See~\cite{hoefler2021sparsity} for an overview of gradient and other sparsification approaches in deep learning. 

However, sparse reductions suffer from scalability issues. Specifically, the communication volume of the existing sparse reduction algorithms grows with the number of processes $P$. Taking the \textit{allgather}-based sparse reduction~\cite{renggli2019sparcml,wang2020fft,shi2019understanding} as an example, its communication volume is proportional to $P$, which eventually surpasses the dense \textit{allreduce} as $P$ increases. Other more complex algorithms~\cite{renggli2019sparcml} suffer from significant fill-in during the reduction, which also leads to a quick increase of the data volume as $P$ grows, and may degrade to dense representations on the fly.
For example, let us assume the model has 1 million weights and it is 99\% sparse at each node---thus, each node contributes its 10,000 largest gradient values and their indexes to the calculation.
Let us now assume that the computation is distributed across 128 data-parallel nodes and the reduction uses a dissemination algorithm~\cite{hensgen1988two,li2013numa} with 7 stages.
In stage one, each process communicates its 10,000 values to be summed up. 
Each process now enters the next stage with up to 20,000 values. Those again are summed up leading to up to 40,000 values in stage 3 (if the value indexes do not overlap). 
The number of values grows exponentially until the algorithm converges after 7 stages with 640,000 values (nearly dense!).
Even with overlapping indexes, the fill-in will quickly diminish the benefits of gradient sparsity in practice and lead to large and suboptimal communication volumes~\cite{renggli2019sparcml}.

We show how to solve or significantly alleviate the scalability issues for large allreduce operations, leading to an asymptotically optimal O($k$) sparse reduction algorithm.
Our intuitive and effective scheme, called O$k$-Top$k$ is easy to implement and can be extended with several features to improve its performance:
(1) explicit sparsity load balancing can distribute the communication and computation more evenly, leading to higher performance; 
(2) a shifted schedule and bucketing during the reduction phase avoids local hot-spots; 
and (3) an efficient selection scheme for top-$k$ values avoids costly sorting of values leading to a significant speedup.

We implement O$k$-Top$k$ in PyTorch~\cite{paszke2019pytorch} and compare it to four other sparse allreduce approaches.
Specifically, O$k$-Top$k$ enables:

\begin{itemize}
  \item a novel sparse allreduce incurring less than 6$k$ (asymptotically optimal) communication volume which is more scalable than the existing algorithms,
  \item a parallel SGD optimizer using the proposed sparse allreduce with high training speed and convergence guarantee,
  \item an efficient and accurate top-$k$ values prediction by regarding the gradient values (along the time dimension) as a slowly changing stochastic process.
\end{itemize}
  
% for performance
We study the parallel scalability and the convergence of different neural networks, including image classification (VGG-16~\cite{simonyan2014very} on Cifar-10), speech recognition (LSTM~\cite{hochreiter1997long} on AN4), and natural language processing (BERT~\cite{devlin2018bert} on Wikipedia), on the Piz Daint supercomputer with a Cray Aries HPC network. Compared with the state-of-the-art approaches, O$k$-Top$k$ achieves the fastest time-to-solution (i.e., reaching the target accuracy/score using the shortest time for full training), and significantly improves training throughput (e.g., 3.29x-12.95x improvement for BERT on 256 GPUs).
We expect speedups to be even bigger in cloud environments with commodity networks. The code of O$k$-Top$k$ is available:  \url{https://github.com/Shigangli/Ok-Topk}

\section{Background and Related Work}
\label{sec:bg}

Mini-batch stochastic gradient descent (SGD)~\cite{bottou2018optimization} is the mainstream method to train deep neural networks. Let $b$ be the mini-batch size, $w_t$ the neural network weights at iteration $t$, $(x_i, y_i)$ a sample in a mini-batch, and $\ell$ a loss function. During training, it computes the loss in the forward pass for each sample as $\ell(w_t, x_i, y_i)$, and then a stochastic gradient in the backward pass as
\[ G_t(w_t) = \frac{1}{b} \sum_{i = 0}^b \nabla \ell(w_t, x_i, y_i). \] The model is trained in iterations such that $w_{t+1} = w_t - \alpha G_t(w_t)$, where $\alpha$ is the learning rate.

To scale up the training process to parallel machines, \textit{data parallelism}~\cite{goyal2017accurate, sergeev2018horovod, you2018imagenet, you2019largebert, li2020taming, li2020breaking} is the common method, in which the mini-batch is partitioned among $P$ workers and each worker maintains a copy of the entire model. Gradient accumulation across $P$ workers is often implemented using a standard dense \textit{allreduce}~\cite{chan2007collective}, leading to about 2$n$ communication volume where $n$ is the number of gradient components (equal to the number of model parameters). However, recent deep learning models~\cite{devlin2018bert, real2019regularized, radford2019language, brown2020language} scale rapidly from millions to billions of parameters, and the proportionally increasing overhead of dense allreduce becomes the main bottleneck in data-parallel training.

Gradient sparsification~\cite{aji2017sparse,alistarh2018convergence,cai2018long,renggli2019sparcml,shi2019understanding,shi2019distributed,shi2019convergence,han2020adaptive,fei2021efficient,xu2021deepreduce} is a key approach to lower the communication volume. By top-$k$ selection, i.e., only selecting the largest (in terms of the absolute value) $k$ of $n$ components, the gradient becomes very sparse (commonly around 99\%). Sparse gradients are accumulated across $P$ workers using a sparse allreduce. Then, the accumulated sparse gradient is used in the Stochastic Gradient Descent (SGD) optimizer to update the model parameters, which is called Top$k$ SGD. The convergence of Top$k$ SGD has been theoretically and empirically proved~\cite{alistarh2018convergence,renggli2019sparcml,shi2019understanding}. However, the parallel scalablity of the existing sparse allreduce algorithms is limited, which makes it very difficult to obtain real performance improvement, especially on the machines (e.g., supercomputers) with high-performance interconnected networks~\cite{shanley2003infiniband,alverson2012cray,slingshot,foley2017ultra}.

\begin{table}[h!]
  \caption{Communication overhead of dense and sparse allreduces ($n$ is the number of gradient components and $n\gg k$)}
  \label{tab:summary}
  \centering
  \renewcommand{\arraystretch}{1.2}
  \begin{threeparttable}
  \begin{tabular}{p{2.1cm}p{3.3cm}p{2cm}}
    \toprule
    Algorithms &  Bandwidth & Latency \\
    \midrule
    Dense~\cite{chan2007collective} &  $2n \frac{P-1}{P} \beta$ & $2(\log_{}P)$$\alpha$ \\
    Top$k$A~\cite{renggli2019sparcml,wang2020fft} &  2$k$($P$-1)$ \beta$ & $(\log_{}P)$$\alpha$ \\
    Top$k$DSA~\cite{renggli2019sparcml} &  $[4k\frac{P-1}{P} \beta$, ($2k+n$)$\frac{P-1}{P} \beta ]$\tnote{1} & ($P + 2\log_{}P$)$ \alpha$ \\
    gTop$k$~\cite{shi2019distributed} &  4$k (\log_{}P)$$\beta$ & 2$(\log_{}P)$$\alpha$ \\
    Gaussian$k$~\cite{shi2019understanding} &  2$k$($P$-1)$ \beta$ & $2(\log_{}P)$$\alpha$ \\
    O$k$-Top$k$ &  $[2k \frac{P-1}{P} \beta$, $6k \frac{P-1}{P} \beta]$\tnote{1} & ($2P + 2\log_{}P$)$\alpha$ \\
    \bottomrule
  \end{tabular}
  \begin{tablenotes}
       \footnotesize
       \item[1] Intervals.
    \end{tablenotes}
  \end{threeparttable}
\end{table}

Table~\ref{tab:summary} summarizes the existing dense and sparse allreduce approaches. We assume all sparse approaches use the coordinate (\textit{COO}) format to store the sparse gradient, which consumes 2$k$ storage, i.e., $k$ values plus $k$ indexes. There are other sparse formats (see~\cite{hoefler2021sparsity} for an overview), but format selection for a given sparsity is not the topic of this work. To model the communication overhead, we assume bidirectional and direct point-to-point communication between the compute nodes, and use the classic latency-bandwidth cost model. The cost of sending a message of size $L$ is $\alpha + \beta L$, where $\alpha$ is the latency and $\beta$ is the transfer time per word. 

For \textbf{dense} allreduce, Rabenseifner’s algorithm~\cite{chan2007collective} reaches the lower bound~\cite{thakur2005optimization} on the bandwidth term (i.e., about 2$n$). \textbf{Top$k$A} represents the \textit{A}llgather based approach~\cite{renggli2019sparcml,shi2019distributed}, in which each worker gathers the sparse gradients across $P$ workers, and then the sparse reduction is conducted locally. Although Top$k$A is easy to realize and does not suffer from the fill-in problem, the bandwidth overhead of allgather is proportional to $P$~\cite{thakur2005optimization,chan2007collective} and thus not scalable. \textbf{Top$k$DSA} represents the \textit{D}ynamic \textit{S}parse \textit{A}llreduce used in SparCML~\cite{renggli2019sparcml}, which consists of reduce-scatter and allgather (motivated by Rabenseifner’s algorithm) on the sparse gradients. In the best case of that the indexes of top-$k$ values are fully overlapped across $P$ workers and the top-$k$ values are uniformly distributed in the gradient space, Top$k$DSA only incurs about 4$k$ communication volume. However, 
the best case is almost never encountered in the real world of distributed training, and Top$k$DSA usually suffers from the fill-in problem and switches to a dense allgather before sparsity cannot bring benefit, which incurs about 2$k$+$n$ communication volume. \textbf{gTop$k$}~\cite{shi2019distributed} implements the sparse allreduce using a reduction tree followed by a broadcast tree. To solve the fill-in problem, gTop$k$ hierarchically selects top-$k$ values in each level of the reduction tree, which results in 4$k \log_{}P$ communication volume. \textbf{Gaussian$k$}~\cite{shi2019understanding} uses the same sparse allreduce algorithm as Top$k$A with a further optimization for top-$k$ selection.
For our \textbf{O$k$-Top$k$}, the communication volume is bounded by 6$k$. Although O$k$-Top$k$ has a little higher latency term than the others, we target large-scale models and thus the bandwidth term dominates. Since the bandwidth term of O$k$-Top$k$ is only related to $k$, our algorithm is more efficient and scalable than all others. See Section~\ref{sec:casestudies} for experimental results.

Gradient quantization~\cite{dryden2016communication,alistarh2017qsgd,10.5555/3294771.3294915,horvath2019stochastic,nadiradze2021asynchronous,bernstein2018signsgd}, which reduces the precision and uses a smaller number of bits to represent gradient values, is another technique to reduce the communication volume. Note that this method is orthogonal to gradient sparsification. A combination of sparsification and quantization is studied in SparCML~\cite{renggli2019sparcml}.

Another issue is the sparsification overhead. Although gradient sparsification significantly reduces the local message size, top-$k$ selection on many-core architectures, such as GPU, is not efficient. A native method is to sort all values and then select the top-$k$ components. Asymptotically optimal comparison-based sorting algorithms, such as merge sort and heap sort, have O($n\log_{}n$) complexity, but not friendly to GPU. Bitonic sort is friendly to GPU but requires O($n\log^{2}_{}n$) comparisons. The quickselect~\cite{mahmoud1995analysis} based top-$k$ selection has an average complexity of O($n$) but again not GPU-friendly. Bitonic top-$k$~\cite{shanbhag2018efficient} is a GPU-friendly algorithm with complexity O($n\log^{2}_{}k$), but still not good enough for large $k$. To lower the overhead of sparsification, Gaussian$k$~\cite{shi2019understanding} approximates the gradient values distribution to a Gaussian distribution with the same mean and standard deviation, and then estimates a threshold using the percent-point function and selects the values above the threshold. The top-$k$ selection in Gaussian$k$ is GPU-friendly with complexity O($n$), but it usually underestimates the value of $k$ because of the difference between Gaussian and the real distributions (see Section~\ref{para:topkselect}). Adjusting the threshold adaptively (e.g., lower the threshold for an underestimated $k$)~\cite{shi2019understanding} is difficult to be accurate. In O$k$-Top$k$, we use a different method for the top-$k$ selection. We observe that the distribution of gradient values changes slowly during training. Therefore, we periodically calculate the accurate threshold and reuse it in the following iterations within a period. Empirical results show that this threshold reuse strategy achieves both accuracy (see Section~\ref{sec:tpkpre}) and efficiency (see Section~\ref{sec:casestudies}) when selecting local and global top-$k$ values in O$k$-Top$k$.

\begin{figure}[ht!]
\centering\includegraphics[width=0.9\linewidth]{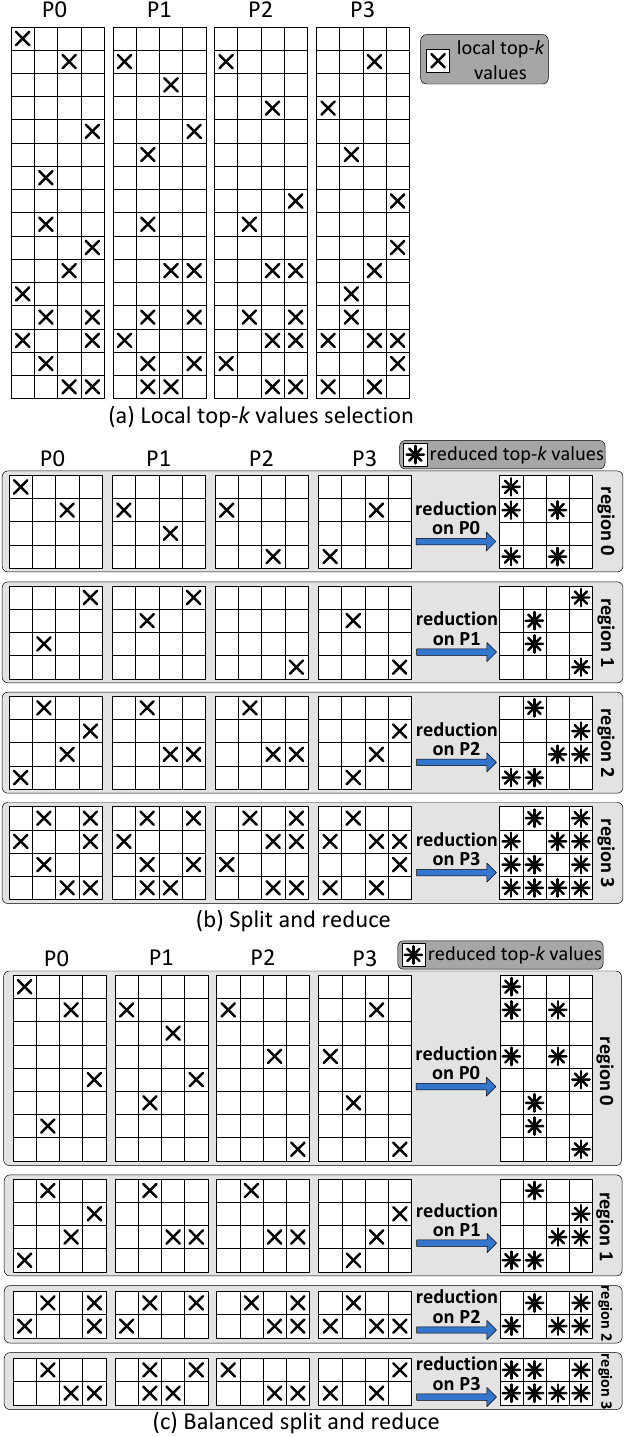}
\caption{\label{splitreduce} Gradient space split with balanced top-$k$ values and reduction on subspace.}
\end{figure}

\section{O($k$) Sparse Allreduce}

In this section, we will present the sparse allreduce algorithm of O$k$-Top$k$, analyze the complexity using the aforementioned latency ($\alpha$) - bandwidth ($\beta$) cost model, and prove its optimality. We use \textit{COO} format to store the sparse gradient. Since the algorithm incurs less than 6$k$ communication volume, we call it \textit{O($k$) sparse allreduce}.

\subsection{Sparse allreduce algorithm}

O($k$) sparse allreduce mainly includes two phases: (1) \textit{split and reduce}, and (2) \textit{balance and allgatherv}. During the two phases, we propose to use an efficient top-$k$ selection strategy to select (what we call) \textit{local} top-$k$ values and \textit{global} top-$k$ values, respectively. Specifically, the semantic of O($k$) sparse allreduce is defined by
$\mbox{Top}k (\sum_{i = 0}^{P-1} \mbox{Top}k (G^i_t))$,
where $G^i_t$ is the sparse gradient on worker $i$ at training iteration $t$, the inner Top$k$ operator is the local top-$k$ selection, and the outer Top$k$ operator is the global top-$k$ selection.

\subsubsection{Split and reduce} 
\label{para:splitreduce}
Figure~\ref{splitreduce} presents the \textit{split and reduce} phase. Suppose we have 4 workers and each worker has a 2D gradient of size 16x4. In Figure~\ref{splitreduce}(a), each worker selects the local top-$k$ values to sparsify the gradient. How to efficiently select the top-$k$ values will be discussed in Section~\ref{para:topkselect}. Then, a straightforward \textit{split and reduce} for the sparse gradients is presented in Figure~\ref{splitreduce}(b), in which the 2D space of the sparse gradient is evenly partitioned into $P$ regions and worker $i$ is responsible for the reduction on region $i$. Each worker $i$ receives sparse regions from the other workers and then conducts the reduction locally. However, this simple partitioning method may lead to severe load imbalance, since the top-$k$ values may not be evenly distributed among the regions. In an extreme case, all local top-$k$ values will be in region 0 of each worker, then worker 0 has to receive a total of $2(P-1)k$ elements (i.e., values and indexes) while the other workers receive zero elements.

Without loss of generality, we can make a more balanced partition (as shown in Figure~\ref{splitreduce}(c)) based on our observations for deep learning tasks: the coordinates distribution of the local top-$k$ values of the gradient is approximately consistent among the workers at the coarse-grained (e.g., region-wise) level, and changes slowly during training. To achieve a balanced partition, each worker calculates the local boundaries of the $P$ regions by balancing the local top-$k$ values. Then, a consensus is made among $P$ workers by globally averaging the $P$-dimensional boundary vectors, which requires an allreduce with message size of $P$ elements. The boundaries are recalculated after every $\tau$ iterations. We empirically set $\tau = 64$ to get a performance benefit from periodic space repartition as shown in Section~\ref{sec:evalbalance}. Note that the small overhead of allreduce (i.e., $(\log_{}P)\alpha$) is amortized by the reuse in the following $\tau$-1 iterations, resulting in only $(\log_{}P)\alpha/\tau$ overhead per iteration. Therefore, the overhead of boundary recalculation can be ignored. After making a balanced split, each worker approximately receives 2$k/P$ elements from any of the other $P-1$ workers. Therefore, the overhead is
\begin{equation} 
C_{split\_and\_reduce} = (P-1) \alpha + 2k \frac{P-1}{P} \beta
\end{equation}

We further make two optimizations for \textit{split and reduce}, including destination rotation and bucketing. As shown in Figure~\ref{splitreduceopt}(a), a native communication pattern is that all workers send data to worker $i$ at step $i$, which may lead to endpoint congestion~\cite{wu2019network}. To avoid these hot-spots, we rotate the destinations of each worker as shown in Figure~\ref{splitreduceopt}(b). Furthermore, to utilize the network parallelism, we bucketize the communications. The messages with in a bucket are sent out simultaneously using non-blocking point-to-point communication functions. Communications in the current bucket can be overlapped with the computation (i.e., local reduction) of the previous bucket.

\begin{figure}[ht!]
\centering\includegraphics[width=0.85\linewidth]{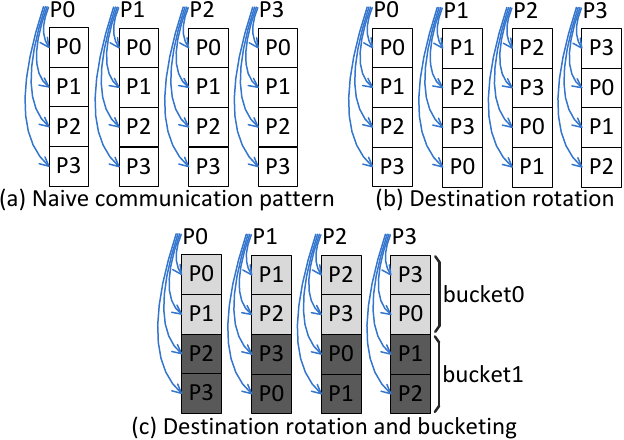}
\caption{\label{splitreduceopt} Top-$k$ values reduction with endpoint congestion avoidance and network parallelism exploitation.}
\end{figure}

\subsubsection{Balance and allgatherv}
\label{para:ballgatherv}
Figure~\ref{allgatherv} presents the phase of \textit{balance and allgatherv}. First, each worker selects the global top-$k$ values from the reduced top-$k$ values in the region that the worker is in charge of. Note that the global top-$k$ selection only happens locally according to an estimated threshold (will be discussed in detail in Section~\ref{para:topkselect}). Next, each worker packages the selected global top-$k$ values (and the corresponding indexes) into a consecutive buffer. Similar to the local top-$k$ values, the global top-$k$ values may also not be evenly distributed among the workers, causing load imbalance. In an extreme case, all global top-$k$ values will be in one worker. The classic recursive doubling based allgatherv~\cite{thakur2005optimization} would incur $2k\log_{}P$ communication volume, namely, total $\log_{}P$ steps with each step causing $2k$ traffic. 

\begin{figure}[ht!]
\centering\includegraphics[width=0.93\linewidth]{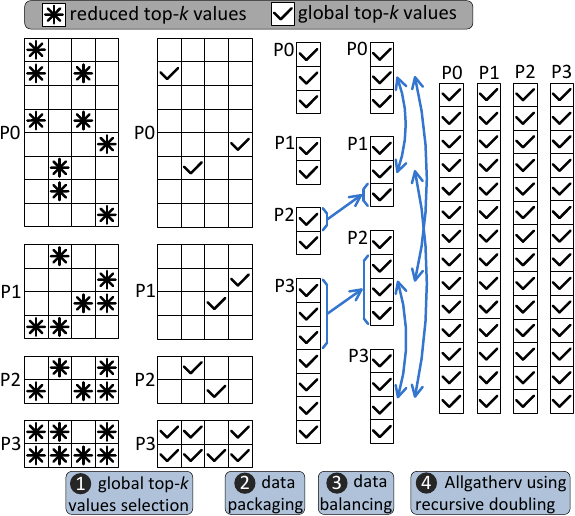}
\caption{\label{allgatherv} Data balancing and allgatherv for global top-$k$ values.}
\end{figure}

To bound the communication volume by 4$k$, we conduct a data balancing step after packaging. Before balancing, an allgather is required to collect the consecutive buffer sizes from $P$ workers, which only incurs $(\log_{}P)\alpha$ overhead (the bandwidth term can be ignored). Then, each worker uses the collected buffer sizes to generate the communication scheme (i.e., which data chunk should be sent to which worker) for data balancing. We use point-to-point communication (blue arrows in the step of \textit{data balancing}) to realize the scheme. The overhead of data balancing is bounded by the extreme case of all global top-$k$ values locate in one worker, where data balancing costs $P \alpha + 2k \frac{P-1}{P}\beta$. Data balancing in any other case has less cost than the extreme case. At last, an allgatherv using recursive doubling on the balanced data has the overhead of $(\log_{}P) \alpha + 2k \frac{P-1}{P}\beta$. Therefore, the overhead of \textit{balance and allgatherv} is

\begin{equation} 
C_{balance\_and\_allgatherv} \leq (P + 2\log_{}P) \alpha + 4k \frac{P-1}{P} \beta .
\end{equation}

By adding the costs of the two phases, the total overhead of O($k$) sparse allreduce is

\begin{equation}
\label{ok:cost}
C_{Ok\_sparse\_allreduce} \leq (2P + 2\log_{}P) \alpha + 6k \frac{P-1}{P} \beta
\end{equation}

\subsubsection{Efficient selection for top-$k$ values}
\label{para:topkselect}

O$k$-Top$k$ relies on estimated thresholds to approximately select the local and global top-$k$ values. The key idea is to regard the gradient values (along the time dimension) as a slowly changing stochastic process $G(t)$. Specifically, the statistics (such as top-$k$ thresholds) of $G(t), G(t+1), ... G(t+\tau^{\prime}-1)$ change very slowly. Therefore, we only calculate the accurate thresholds for local and global top-$k$ values after every $\tau^{\prime}$ iterations, and then reuse the thresholds in the following $\tau^{\prime}$-1 iterations. The accurate threshold can be obtained by sorting the gradient values and returning the $k$-th largest value. Top-$k$ selection according to a threshold only requires $n$ comparisons and is quite efficient on GPU. The overhead of accurate threshold calculation is amortized by the reuse.

\begin{figure}[!ht]
  \centering
  \begin{subfigure}{\linewidth}
    \centering
    \includegraphics[width=.9\linewidth]{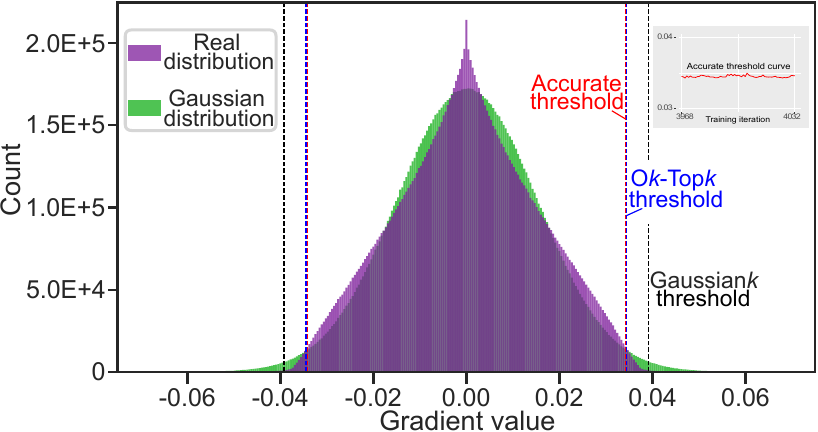}
    \caption{VGG-16~\cite{simonyan2014very} on Cifar-10 with density = 1.0\%.}
    \label{vgggraddist}
  \end{subfigure}
  
  \begin{subfigure}{\linewidth}
    \centering
    \includegraphics[width=.9\linewidth]{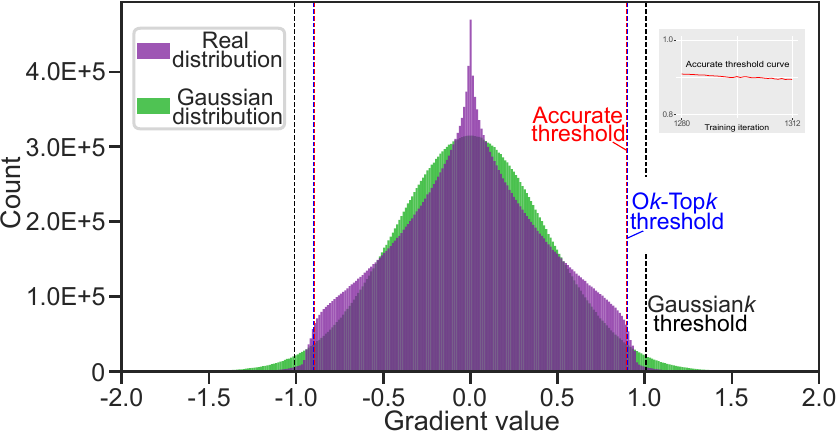}
    \caption{LSTM~\cite{hochreiter1997long} on AN4~\cite{acero1990environmental} with density = 2.0\%.}
    \label{lstmgraddist}
  \end{subfigure}
  
  \begin{subfigure}{\linewidth}
    \centering
    \includegraphics[width=.9\linewidth]{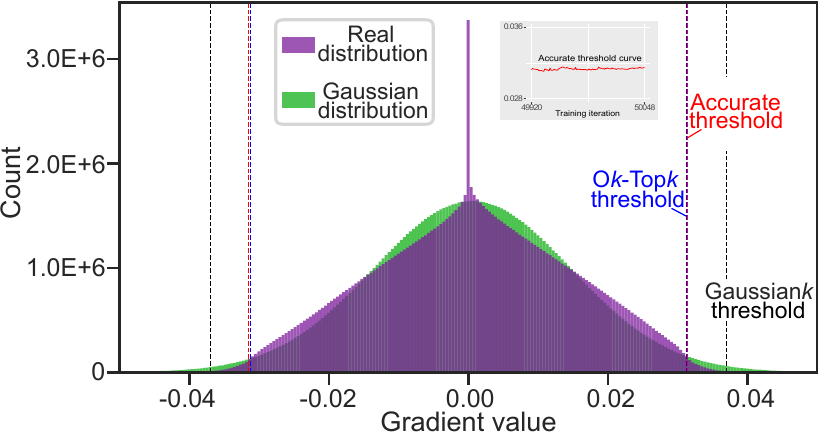}
    \caption{BERT~\cite{devlin2018bert} on Wikipedia~\cite{devlin2018bert} with density = 1.0\%.}
    \label{bertgraddist}
  \end{subfigure}
  \caption{Gradient value distribution and local top-$k$ threshold predictions. Density is defined as $k/n$.}
  \label{fig:graddist}
\end{figure}

We validate our claim by the empirical results from different deep learning tasks presented in Figure~\ref{fig:graddist}. The gradient value distribution shown in Figure~\ref{fig:graddist} is for a selected iteration where O$k$-Top$k$ uses a threshold calculated more than 25 iterations ago. We can see the threshold of O$k$-Top$k$ is still very close to the accurate threshold (see Section~\ref{sec:tpkpre} for the accuracy verification for top-$k$ selections of O$k$-Top$k$ in the scenario of full training). On the contrary, Gaussian$k$ severely underestimates the value of $k$ by predicting a larger threshold, especially after the first few training epochs. This is because as the training progresses, the gradient values are getting closer to zero. Gaussian distribution, with the same mean and standard deviation as the real distribution, usually has a longer tail than the real distribution (see Figure~\ref{fig:graddist}).

\begin{algorithm}[h!]
  %\footnotesize
  \begin{algorithmic}[1]
    \State \textbf{Inputs:} stochastic gradient $G^i_t$ at worker $i$, iteration $t$ ($t$>0), value $k$, space repartition period $\tau$, thresholds re-evaluation period $\tau^{\prime}$.
    \If{($t$-1)$\mod\tau^{\prime} == 0$}
      \State $local$\_$th$ = \textcolor{blue}{$th$\_$re$\_$evaluate$}($G^i_t$, $k$) 
    \EndIf
    \If{($t$-1)$\mod\tau == 0$}
      \State $boundaries$ = \textcolor{blue}{$space\_repartition$}($G^i_t$, $local\_th$)
    \EndIf
    \State $reduced\_topk^{i}$, $local\_topk\_indexes$ =$\quad \quad \quad \quad \quad \quad \quad $ \textcolor{blue}{\bm{$split\_and\_reduce$}}($G^i_t$, $local\_th$, $boundaries$)
    
    \If{($t$-1)$\mod\tau^{\prime} == 0$}
      \State $all\_reduced\_topk$ = \textcolor{blue}{$allgatherv$}($reduced\_topk^{i}$)
      \State $global$\_$th$ = \textcolor{blue}{$th$\_$re$\_$evaluate$}($all\_reduced\_topk$, $k$) 
    \EndIf
    \State $u_{t}$, $global\_topk\_indexes$ = $\quad \quad \quad \quad \quad \quad \quad \quad \quad \quad $ \textcolor{blue}{\bm{$balance\_and\_allgatherv$}}($reduced\_topk^{i}$, $global$\_$th$)
    \State $indexes$ = $local\_topk\_indexes \cap global\_topk\_indexes$
    \State \Return $u_{t}$, $indexes$
  \end{algorithmic}
  \caption{O($k$) sparse allreduce}
  \label{alg:sparseallreduce}
\end{algorithm}

\subsubsection{Pseudocode of O($k$) sparse allreduce} We present the pseudocode of O($k$) sparse allreduce in Algorithm~\ref{alg:sparseallreduce}. In Lines 2-4, the local top-$k$ threshold is re-evaluated after every $\tau^{\prime}$ iterations. In Lines 5-7, the region boundaries are re-evaluated after every $\tau$ iterations. \textit{Split and reduce} is conducted in Line 8, which returns the reduced local top-$k$ values in region $i$ and the indexes of local top-$k$ values. In Lines 9-12, the global top-$k$ threshold is re-evaluated after every $\tau^{\prime}$ iterations. \textit{Balance and allgatherv} is conducted in Line 13, which returns a sparse tensor $u_t$ with global top-$k$ values and the indexes of global top-$k$ values. Line 14 calculates the intersection of the indexes of local top-$k$ values and the indexes of global top-$k$ values. This intersection (will be used in O$k$-Top$k$ SGD in Section~\ref{sec:oktopksgd}) covers the indexes of local values which eventually contribute to the global top-$k$ values.

\subsection{Lower bound for communication volume}

\begin{theorem}
For sparse gradients stored in COO format, O($k$) sparse allreduce incurs at least $2k \frac{P-1}{P}$ communication volume.
\label{thm:lowerb}
\end{theorem}

\begin{proof}
For O($k$) sparse allreduce, each worker eventually obtains the global top-$k$ values. Assume that all workers receive less than $k \frac{P-1}{P}$ values from the others, which means that each worker already has more than $\frac{k}{P}$ of the global top-$k$ values locally. By adding up the number of global top-$k$ values in each worker, we obtain more than $k$ global top-$k$ values, which is impossible. Therefore, each worker has to receive at least $k \frac{P-1}{P}$ values. Considering the corresponding $k$ indexes, the lower bound is $2k \frac{P-1}{P}$.

\end{proof}

The lower bound in Theorem~\ref{thm:lowerb} is achieved by O($k$) sparse allreduce in the following special case: All local top-$k$ values of worker $i$ are in region $i$ that worker $i$ is in charge of, so that the reduction across $P$ workers is no longer required. Furthermore, the global top-$k$ values are uniformly distributed among $P$ workers, namely each worker has exactly $\frac{k}{P}$ of the global top-$k$ values. Then, an allgather to obtain the global top-$k$ values (plus $k$ indexes) incurs $2k \frac{P-1}{P}$ communication volume. Therefore, the lower bound is tight. Since O($k$) sparse allreduce incurs at most $6k \frac{P-1}{P}$ communication volume (see Equation~\ref{ok:cost}), it is asymptotically optimal.

% \begin{corollary}
% Since O($k$) sparse allreduce incurs at most $6k \frac{P-1}{P}$ communication volume, it is asymptotically optimal.
% \end{corollary}

\begin{algorithm}[h!]
  %\footnotesize
  \begin{algorithmic}[1]
    \State \textbf{Inputs:} stochastic gradient $G^i$($\cdot$) at worker $i$, value $k$, learning rate $\alpha$.
    \State Initialize $\epsilon_{0}^{i}$ = 0, $G_{0}^i$ = 0
    \For{$t = 1$ \textbf{to} $T$}
      \State $acc_{t}^{i}$ = $\epsilon_{t-1}^{i} + \alpha G^i_{t-1}$($w_{t-1}$) \Comment{Accumulate residuals}
      \State $u_{t}$, $indexes$ = \textcolor{blue}{\bm{$Ok\_sparse\_allreduce$}}($acc_{t}^{i}$, $t$, $k$)
      \State $\epsilon_{t}^{i} = acc_{t}^{i} - acc_{t}^{i}$($indexes$) \Comment{Update residuals}
      \State $w_{t} = w_{t-1} - \frac{1}{P}u_{t}$ \Comment{Apply the model update}
    \EndFor
  \end{algorithmic}
  \caption{O$k$-Top$k$ SGD}
  \label{alg:oktopksgd}
\end{algorithm}

\section{O$k$-Top$k$ SGD Algorithm}
\label{sec:oktopksgd}
We discuss how O($k$) sparse allreduce works with the SGD optimizer for distributed training in this section. An algorithmic description of  O$k$-Top$k$ SGD in presented in Algorithm~\ref{alg:oktopksgd}. The key point of the algorithm is to accumulate the residuals (i.e., the gradient values not contributing to the global top-$k$ values) locally, which may be chosen by the top-$k$ selection in the future training iterations to make contribution. Empirical results of existing work~\cite{renggli2019sparcml,shi2019distributed,alistarh2018convergence} show that residual accumulation benefits the convergence. We use $\epsilon_t^i$ to represent the residuals maintained by worker $i$ at iteration $t$. In Line 4 of Algorithm~\ref{alg:oktopksgd}, residuals are accumulated with the newly generated gradient to obtain the accumulator $acc_t^i$. Then, in Line 5, $acc_t^i$ is passed to O($k$) sparse allreduce (presented in Algorithm~\ref{alg:sparseallreduce}), which returns the sparse tensor $u_t$ containing global top-$k$ values and the $indexes$ 
marking the values in $acc_t^i$ which contribute to $u_t$. In Line 6, residuals are updated by setting the values in $acc_t^i$ marked by $indexes$ to zero. In Line 7, $u_t$ is applied to update the model parameters.

\subsection{Convergence proof}
Unless otherwise stated, $\|\cdot\|$ denotes the 2-norm.

\begin{theorem}
Consider the O$k$-Top$k$ SGD algorithm when minimizing a smooth, non-convex objective function $f$. Then there exists a learning rate schedule ($\alpha_{t}$)$_{t=1,T}$ such that the following holds:

\begin{equation} 
\min_{t \in \{ 1,...,T\}} \E[\|\nabla f(w_{t})\|^2] \overset{T \to \infty}{\to} 0
\end{equation}
\label{thm:converge}
\end{theorem}
 
\begin{proof}

The update to $w_{t+1}$ in O$k$-Top$k$ SGD is
\[
\mbox{Top}k \biggl( \frac{1}{P}  \sum_{i = 0}^{P-1} \mbox{Top}k \biggl(\alpha G^i_t( w_t ) + \epsilon^i_t \biggr)\biggr),
\]
while top-$k$ components of the sum of updates across all $P$ workers, i.e., the true global top-$k$  values intended to be applied, is

\[
\mbox{Top}k \biggl( \frac{1}{P} \sum_{i = 0}^{P-1} \biggl(\alpha G^i_t( w_t ) + \epsilon^i_t \biggr) \biggr).
\]

We assume the difference between the update calculated by O$k$-Top$k$ SGD and the true global top-$k$ values is bounded by the norm of the true gradient $G_{t}(w_{t}) = \frac{1}{P}  \sum_{i = 0}^{P-1} G^i_t( w_t )$. Then, we make the following assumption:

\begin{assumption}
There exists a (small) constant $\xi$ such that, for every iteration $t \geq$ 0, we have:
\label{ass:xi}
\end{assumption}

\begin{equation}    
\begin{split}
& \Biggl\| \mbox{Top}k \biggl( \frac{1}{P} \sum_{i = 0}^{P-1} \biggl(\alpha G^i_t( w_t ) + \epsilon^i_t \biggr) \biggr) \\ 
& - \mbox{Top}k \biggl( \frac{1}{P}  \sum_{i = 0}^{P-1} \mbox{Top}k \biggl(\alpha G^i_t( w_t ) + \epsilon^i_t \biggr)\biggr) \Biggr\| \leq \xi\| \alpha G_{t}(w_{t})\|
\end{split} 
\end{equation}

We validate Assumption~\ref{ass:xi} by the empirical results on different deep learning tasks in Section~\ref{sec:xi}. Then, we utilize the convergence proof process for Top$k$ SGD in the non-convex case, presented in the work of ~\cite{alistarh2018convergence}, to prove Theorem~\ref{thm:converge}.

\end{proof}

Regarding Theorem~\ref{thm:converge}, we have the following discussion. First, since we analyze non-convex objectives, a weaker notion of convergence than in the convex case (where we can prove convergence to a global minimum) is settled. Specifically, for a given sequence of learning rates, we prove that the algorithm will converge to a stationary point of negligible gradient. Second, like most theoretical results, it does not provide a precise set for the hyperparameters, except the indication of diminishing learning rates.

\section{Evaluations}
\label{sec:eval}

We conduct our experiments on the CSCS Piz Daint supercomputer. Each Cray XC50 compute node contains an Intel Xeon E5-2690 CPU, and one NVIDIA P100 GPU with 16 GB global memory. We utilize the GPU for acceleration in all following
experiments. The compute nodes of Piz Daint are connected by a Cray Aries interconnect network in a Dragonfly topology.

 \begin{table}[!ht]
	\caption{Neural networks used for evaluation.}
	\label{tab:networks}
	\centering
	\small
	\begin{tabular}{p{2.65cm}p{1.57cm}p{1.35cm}p{1.6cm}}
	%\begin{tabular}{l l c c c }
		\toprule
		Tasks & Models & Parameters & Dataset \tabularnewline
		\midrule
		Image classification & VGG-16~\cite{simonyan2014very} & 14,728,266 & Cifar-10 \tabularnewline
		Speech recognition & LSTM~\cite{hochreiter1997long} & 27,569,568 & AN4~\cite{acero1990environmental} \tabularnewline
		Language processing & BERT~\cite{devlin2018bert} & 133,547,324 & Wikipedia~\cite{devlin2018bert} \tabularnewline
		\bottomrule
	\end{tabular}
\end{table}

We use three neural networks from different deep learning domains summarized in Table~\ref{tab:networks} for evaluation. For VGG-16, we use SGD optimizer with initial learning rate of 0.1; for LSTM, we use SGD optimizer with initial learning rate of 1e-3; for BERT, we use Adam~\cite{kingma2014adam} optimizer with initial learning rate of 2e-4, $\beta_{1} = $0.9, $\beta_{2} = $0.999, weight decay of 0.01, and linear decay of the learning rate. For BERT, sparse allreduce is conducted on the gradients and Adam optimizer is applied afterwards. We compare the performance of O$k$-Top$k$ with the parallel SGD schemes using the dense and sparse allreduce algorithms listed in Table~\ref{tab:summary}, which covers the state-of-the-art. For a fair comparison, all schemes are implemented in PyTorch~\cite{paszke2019pytorch} with \textit{mpi4py} as the communication library, which is built against Cray-MPICH 7.7.16. Commonly, the gradients of network layers locate in non-contiguous buffers. We use \textbf{Dense} to denote a single dense allreduce on a long message aggregated from the gradients of all neural network layers. Furthermore, we use \textbf{DenseOvlp} to denote dense allreduce with the optimization of communication and computation overlap. For DenseOvlp, the gradients are grouped into buckets and the message aggregation is conducted within each bucket; once the aggregated message in a bucket is ready, a dense allreduce is fired. The sparse allreduce counterparts (i.e., \textbf{Top$k$A}, \textbf{Top$k$DSA}, \textbf{gTop$k$}, and \textbf{Gaussian$k$}) are already discussed in Section~\ref{sec:bg}. In all following experiments, we define $density$ as $k/n$, where $n$ is the number of components in the gradient.

We utilize the $topk$ function provided by PyTorch~\cite{paszke2019pytorch}, which is accelerated on GPU, to realize the top-$k$ selection in Top$k$A, Top$k$DSA, and gTop$k$, as well as the periodic threshold re-evaluation in O$k$-Top$k$. 

\begin{figure}[!ht]
  \centering
  \begin{subfigure}{\linewidth}
    \centering
    \includegraphics[width=.9\linewidth]{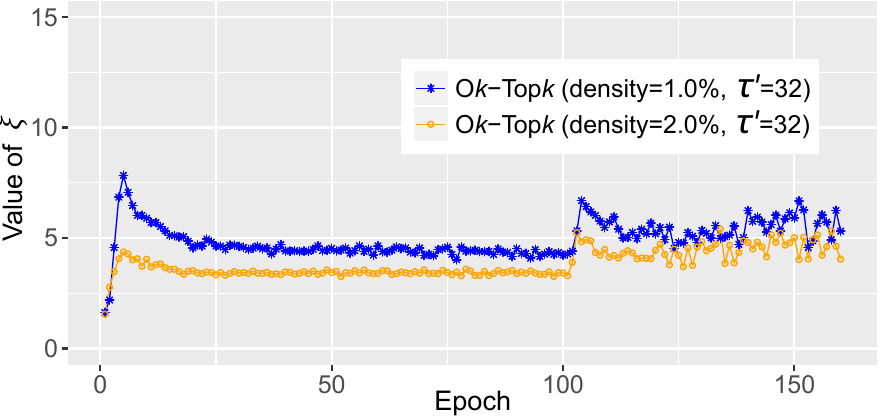}
    \caption{VGG-16 on Cifar-10 running on 16 GPU nodes.}
    \label{vggxi}
  \end{subfigure}
  \begin{subfigure}{\linewidth}
    \centering
  \includegraphics[width=.9\linewidth]{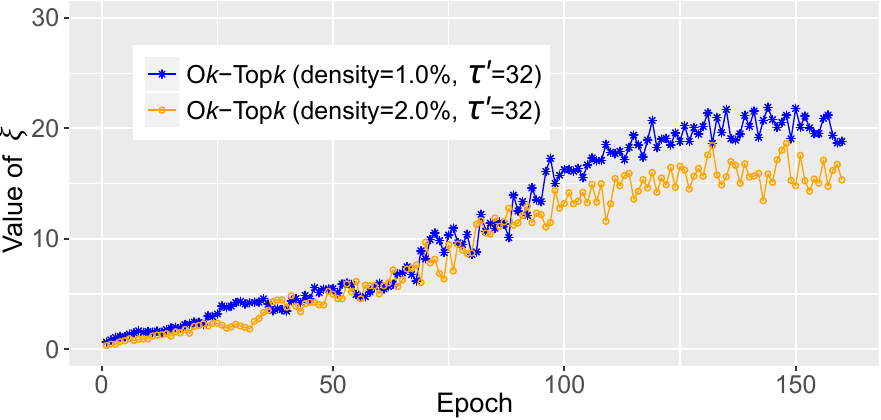}
      \caption{LSTM on AN4 running on 32 GPU nodes.}\label{lstmxi}
  \end{subfigure}
  \begin{subfigure}{\linewidth}
    \centering
  \includegraphics[width=.9\linewidth]{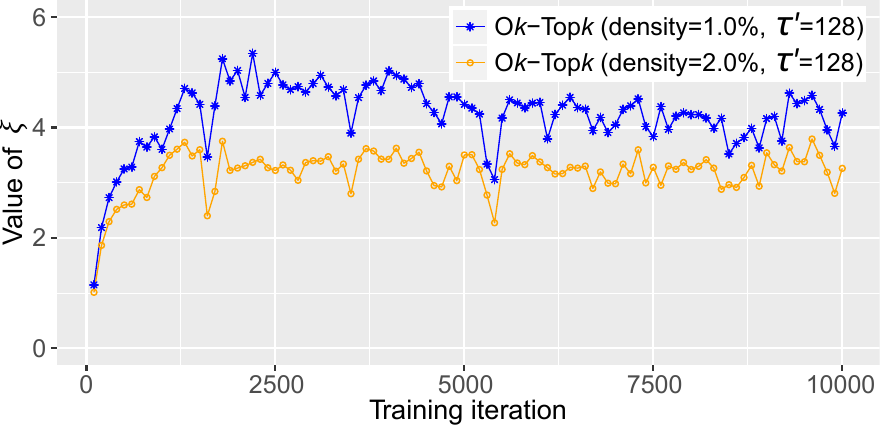}
      \caption{BERT on Wikipedia on 32 GPU nodes.}\label{bertxi}
  \end{subfigure}
  \caption{The empirical value of $\xi$.}
  \label{xivalue}
\end{figure}

\begin{figure}[!ht]

  \begin{subfigure}{\linewidth}
    \centering
    \includegraphics[width=.9\linewidth]{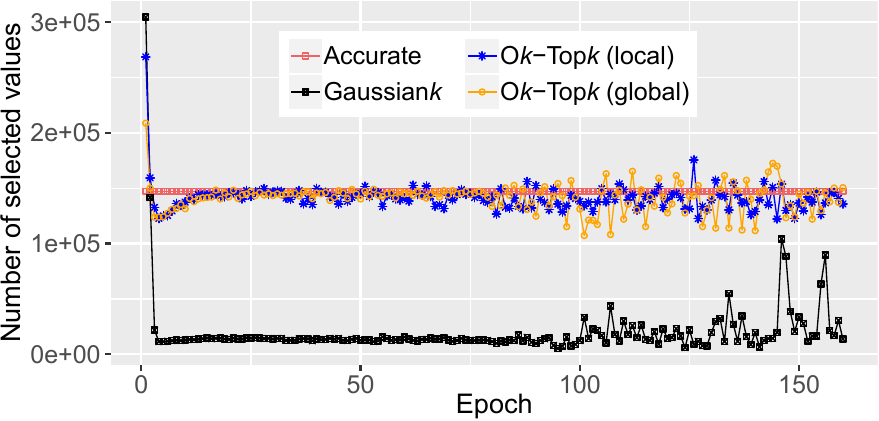}
      \caption{VGG-16 on Cifar-10 on 16 GPUs (density=1.0\%, $\tau^\prime$=32).}\label{vggtopknumber}
  \end{subfigure}
  
  \begin{subfigure}{\linewidth}
    \centering
    \includegraphics[width=.9\linewidth]{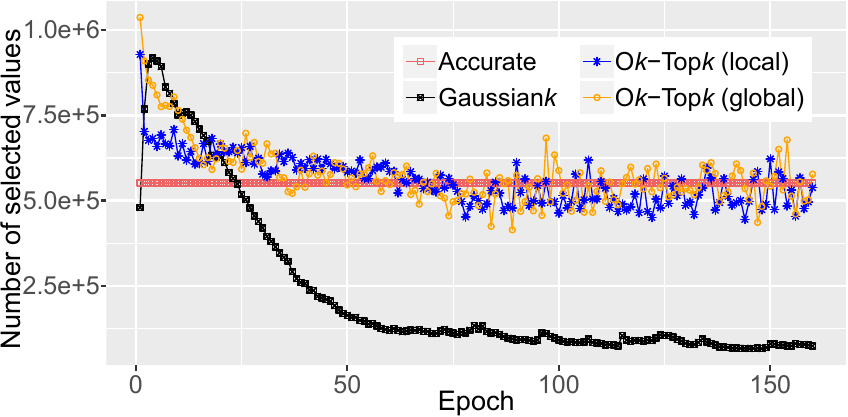}
      \caption{LSTM on AN4 on 32 GPUs (density=2.0\%, $\tau^\prime$=32).}\label{lstmtopknumber}
  \end{subfigure}
  
  \begin{subfigure}{\linewidth}
    \centering
    \includegraphics[width=.9\linewidth]{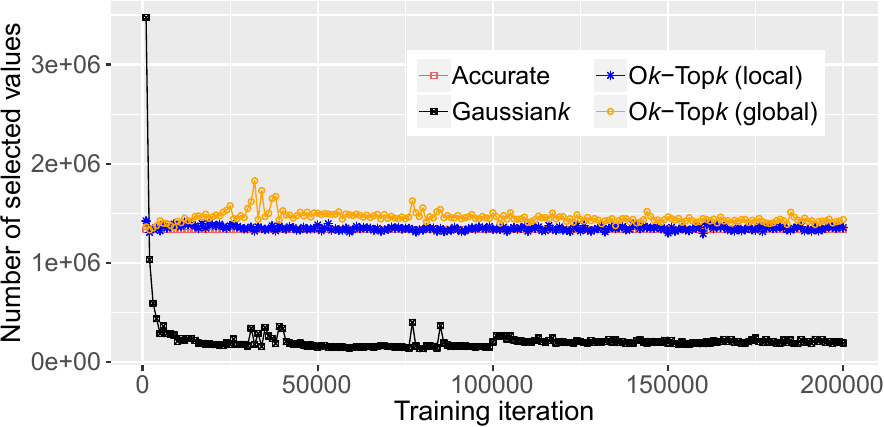}
      \caption{BERT on Wikipedia on 32 GPUs (density=1.0\%, $\tau^\prime$=128).}\label{berttopknumber}
  \end{subfigure}
    \caption{Selections for local and global top-$k$ values}
  \label{fig:topknumbers}
\end{figure}

\subsection{Evaluate the empirical value of $\xi$}
\label{sec:xi}

To validate Assumption~\ref{ass:xi}, we present the empirical values of $\xi$ when training two models until convergence with different densities in Figure~\ref{xivalue}. For VGG-16 and BERT, the value of $\xi$ increases quickly in the first few epochs or training iterations, and then turns to be stable. For LSTM, the value of $\xi$ gradually increases at the beginning and tends to plateau in the second half of the training. For all three models, the value of $\xi$ with a higher density is generally smaller than that with a lower density, especially in the stable intervals. This can be explained trivially by the reason that, the higher the density the higher probability that the results of sparse and dense allreduces get closer. 

As shown in Equation 14 of ~\cite{alistarh2018convergence}, the effect of $\xi$ is dampened by both $P$ and small (i.e., less than 1) learning rates. If $\xi$<$P$ (satisfied by all three models in Figure~\ref{xivalue}) or not too larger than $P$, we consider it has no significant effect on the convergence. Although $\xi$ slightly grows in Figure~\ref{lstmxi}, which is caused by the decreasing of the true gradient norm as the model converges, a small learning rate (e.g., 0.001) can restrain its effect. Overall, Assumption~\ref{ass:xi} empirically holds with relatively low, stable or slowly growing values of $\xi$.

\subsection{Top-$k$ values selection}
\label{sec:tpkpre}

We will verify the accuracy of the top-$k$ selection strategy used by O$k$-Top$k$ on different neural network models. For VGG-16 and LSTM, the models are trained for 160 epochs until convergence with $\tau^\prime$=32. Recall that $\tau^\prime$ is the period of thresholds re-evaluation. For BERT, the model is trained for 200,000 iterations (more than 20 hours on 32 GPU nodes) with $\tau^\prime$=128. The numbers of local and global top-$k$ values selected by O$k$-Top$k$ during training are monitored. We also record the values of $k$ predicted by Gaussian$k$ for comparison. The results are reported in Figure~\ref{fig:topknumbers}. We can see that the numbers of both local and global top-$k$ values selected by O$k$-Top$k$ are very close to the accurate number for a given density, except that O$k$-Top$k$ overestimates the value of $k$ in the early epochs of VGG-16 and LSTM. 
For both local top-$k$ and global top-$k$ on three models, the average deviation from the accurate number is below 11\%. For example, the average deviation for local top-$k$ selection on BERT is only 1.4\%. These results demonstrate the accuracy of the threshold reuse strategy adopted by O$k$-Top$k$. On the contrary, Gaussian$k$ overestimates the value of $k$ in the first few epochs and then severely underestimate $k$ (an order of magnitude lower than the accurate number) in the following epochs. This can be explained by the difference between Gaussian and the real distributions, as discussed in Section~\ref{para:topkselect}. 

As a comparison, we also count the density of the output buffer (i.e., the accumulated gradient) for Top$k$DSA (Top$k$A has the same density), which expands to 13.2\% and 34.5\% on average for VGG-16 (local density = 1.0\%, on 16 GPUs) and LSTM (local density = 2.0\%, on 32 GPUs), respectively. These statistics show the effect of the fill-in issue for Top$k$DSA.

\begin{figure}[!ht]
  \centering
  \includegraphics[width=0.99\linewidth]{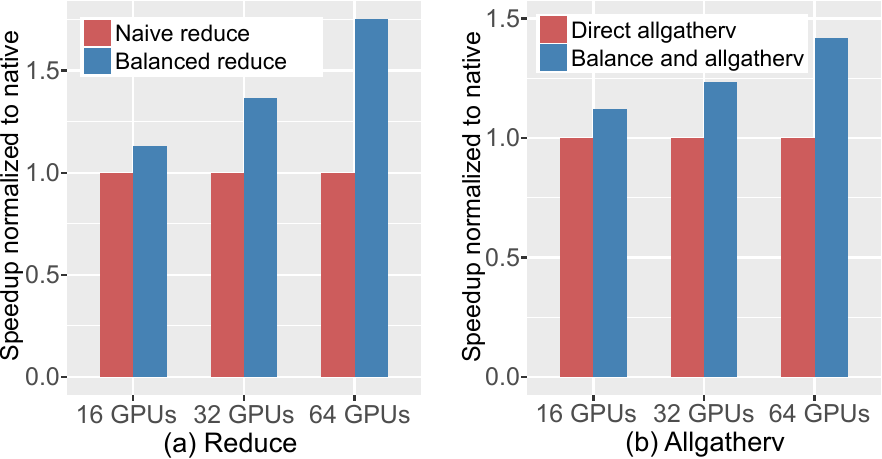}
  \caption{Evaluation for communication balancing in O$k$-Top$k$ using BERT with density = 1.0\%.}
  \label{bertbalance}
\end{figure}

\subsection{Optimizations for load balancing in O$k$-Top$k$}
\label{sec:evalbalance}

To evaluate the effectiveness of the load balancing optimizations of O($k$) sparse allreduce, we train BERT for 8,192 iterations and report the average values. 

First, we evaluate the periodic space repartition strategy (discussed in Section~\ref{para:splitreduce}) for load balancing in the phase of \textit{split and reduce}. The results are presented in Figure~\ref{bertbalance}(a). Recall that we set the period $\tau$ to 64. The repartition overhead is counted and averaged in the runtime of the balanced reduce. In the naive reduce, the gradient space is partitioned into equal-sized regions, regardless of the coordinate distribution of the local top-$k$ values. The balanced reduce achieves 1.13x to 1.75x speedup over the naive one, with a trend of more significant speedup on more GPUs. This trend can be explained by that the load imbalance in the naive reduce incurs up to $2(P-1)k$ communication volume (proportional to $P$). While the balanced reduce incurs less than $2k$ communication volume, which is more scalable.

Next, we evaluate the data balancing strategy (discussed in Section~\ref{para:ballgatherv}) in the phase of \textit{balance and allgatherv}. Although data balancing helps to bound the bandwidth overhead of allgatherv, there is no need to conduct it if the data is roughly balanced already. Empirically, we choose to conduct data balancing before allgatherv if the max data size among $P$ workers is more than four times larger than the average data size, and otherwise use an allgatherv directly. Figure~\ref{bertbalance}(b) presents the results for the iterations where data balancing is triggered. Data balancing and allgatherv achieve 1.12x to 1.43x speedup over the direct allgatherv. For similar reasons as 
\textit{split and reduce}, more speedup is achieved on more GPUs.

\subsection{Case studies on training time and model convergence}
\label{sec:casestudies}

We study the training time and model convergence using real-world applications listed in Table~\ref{tab:networks}. For training time per iteration, we report the average value of full training. To better understand the results, we make a further breakdown of the training time, including sparsification (i.e., top-$k$ selection from the gradient), communication (i.e., dense or sparse allreduces), and computation (i.e., forward and backward passes) plus I/O (i.e., sampling from dataset).

As discussed in Section~\ref{sec:tpkpre}, Gaussian$k$ usually underestimates the value of $k$, which makes the actual density far below the setting. Both empirical and theoretical results~\cite{renggli2019sparcml,shi2019distributed,alistarh2018convergence} show that a very low density would jeopardize the convergence. To make a fair comparison between the counterparts for both training time and model accuracy, we gradually scale the predicted threshold of Gaussian$k$ until the number of selected values is more than 3$k$/4. The threshold adjustment is also suggested by~\cite{shi2019understanding}, although it is difficult to be accurate. The threshold adjustment may slightly increase the sparsification overhead of Gaussian$k$, but compared with the other overheads it can be ignored (see the following results).

\begin{figure}[!t]
  \centering
  \begin{subfigure}{\linewidth}
    \centering
    \includegraphics[width=0.88\linewidth]{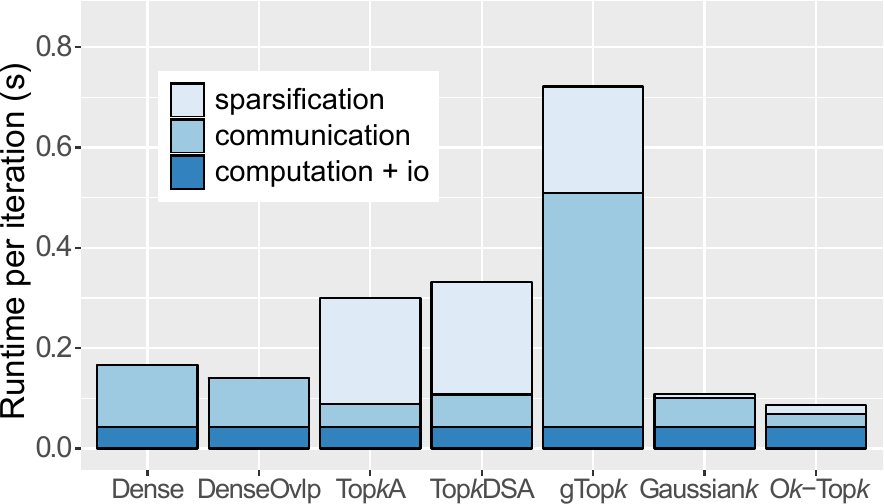}
    \caption{Running on 16 GPU nodes with global batch size = 256.}
    \label{vgg16nodestime}
  \end{subfigure}
  \begin{subfigure}{\linewidth}
    \centering
  \includegraphics[width=0.88\linewidth]{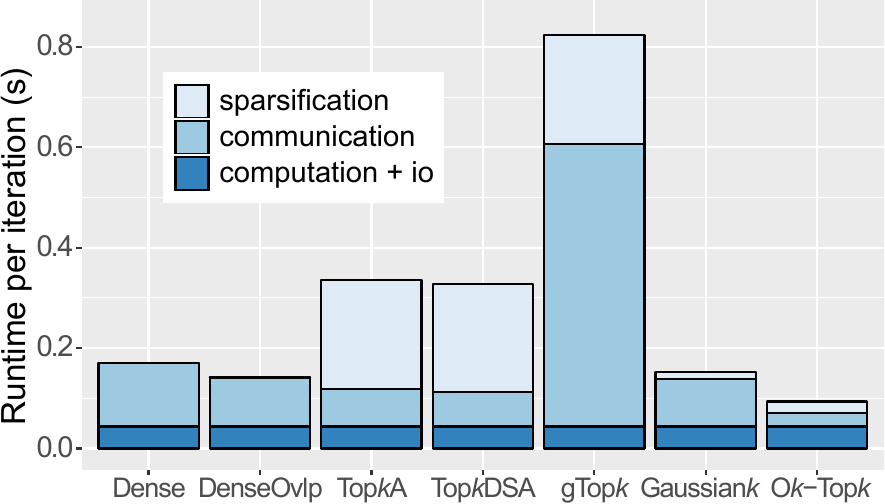}
      \caption{Running on 32 GPU nodes with global batch size = 512.}\label{vgg32nodestime}
  \end{subfigure}
  \caption{Weak scaling of training VGG-16 on Cifar-10 with density = 2.0\%.}
  \label{vggruntime}
\end{figure}

\begin{figure}[!t]
  \centering
  \begin{subfigure}{\linewidth}
    \centering
    \includegraphics[width=0.88\linewidth]{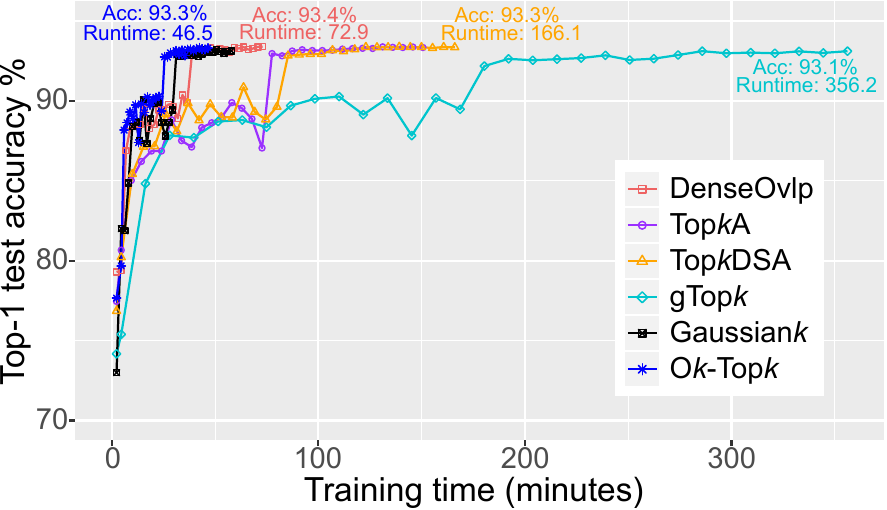}
    \caption{Running on 16 GPU nodes with global batch size = 256.}
    \label{vgg16nodesacc}
  \end{subfigure}
  \begin{subfigure}{\linewidth}
    \centering
  \includegraphics[width=0.88\linewidth]{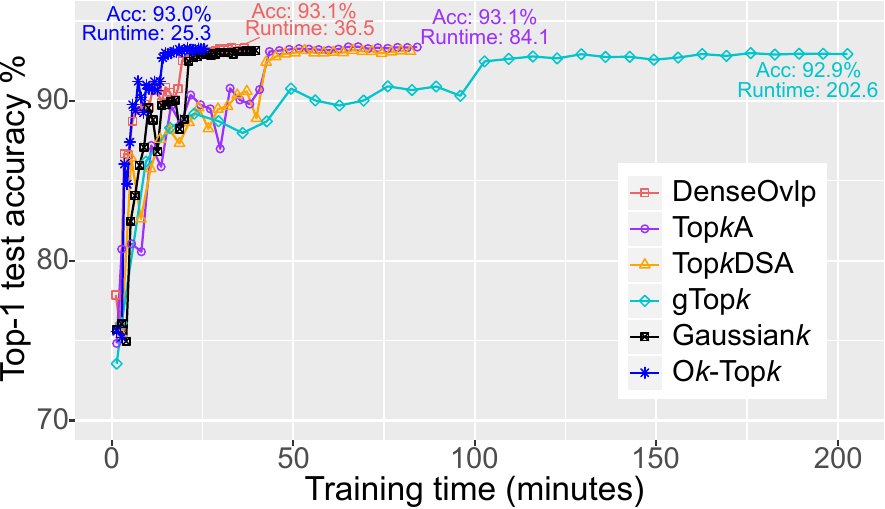}
      \caption{Running on 32 GPU nodes with global batch size = 512.}\label{vgg32nodesacc}
  \end{subfigure}
  \caption{Top-1 test accuracy for VGG-16 on Cifar-10 with density = 2.0\% training for 160 epochs.}
  \label{vggacc}
\end{figure}

\subsubsection{Image classification}

Figure~\ref{vggruntime} presents the results of weak scaling for training VGG-16 on Cifar-10. DenseOvlp outperforms Dense by enabling communication and computation overlap. Although Top$k$A and Top$k$DSA have lower communication overhead than DenseOvlp, they have a high overhead for sparsification, which makes the benefit of lower communication disappear. Note that the communication overhead of gTop$k$ seems much higher than the others; this is because the overhead of hierarchical top-$k$ selections in the reduction-tree (with $\log_{}P$ steps) is also counted in the communication overhead. Among all sparse allreduce schemes, Gaussian$k$ has the lowest sparsification overhead. O$k$-Top$k$ has the lowest communication overhead; by using the threshold reuse strategy, O$k$-Top$k$ only has a slightly higher sparsification overhead than Gaussian$k$. When scaling from 16 GPUs to 32 GPUs, the communication overhead of Top$k$A and Gaussian$k$ almost doubles. This is because allgather-based sparse allreduce is not scalable (see the performance model in Table~\ref{tab:summary}). On 32 GPU nodes, O$k$-Top$k$ outperforms the other schemes by 1.51x-8.83x for the total training time.

\begin{figure}[!h]
  \centering
  \begin{subfigure}{\linewidth}
    \centering
    \includegraphics[width=0.86\linewidth]{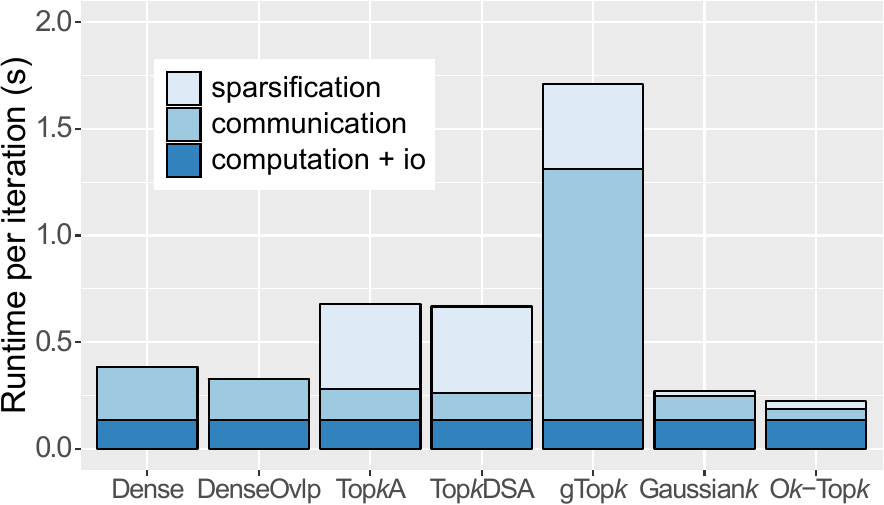}
    \caption{Running on 32 GPU nodes with global batch size = 64.}
    \label{lstm32nodestime}
  \end{subfigure}
  \begin{subfigure}{\linewidth}
    \centering
  \includegraphics[width=0.86\linewidth]{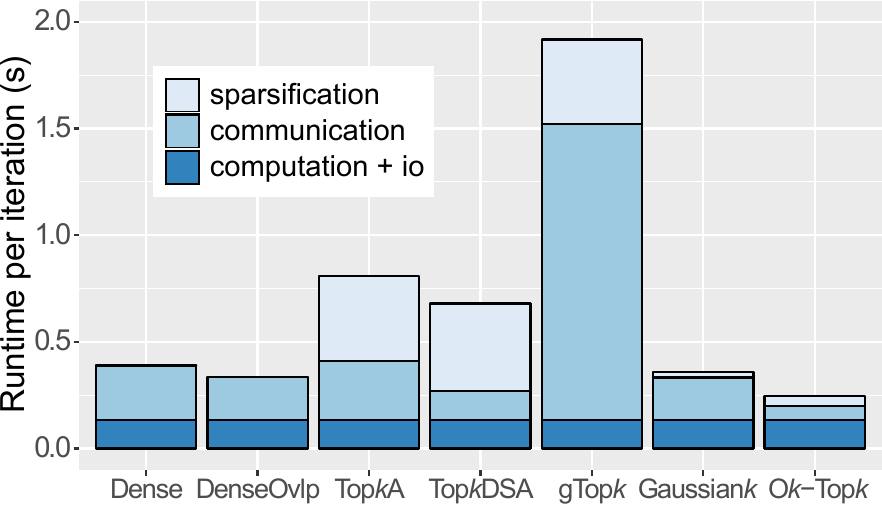}
    \caption{Running on 64 GPU nodes with global batch size = 128.}
    \label{lstm64nodestime}
  \end{subfigure}
  \caption{Weak scaling of training LSTM on AN4 with density = 2.0\%.}
  \label{lstmruntime}
\end{figure}

\begin{figure}[!h]
  \centering
  \begin{subfigure}{\linewidth}
    \centering
    \includegraphics[width=0.86\linewidth]{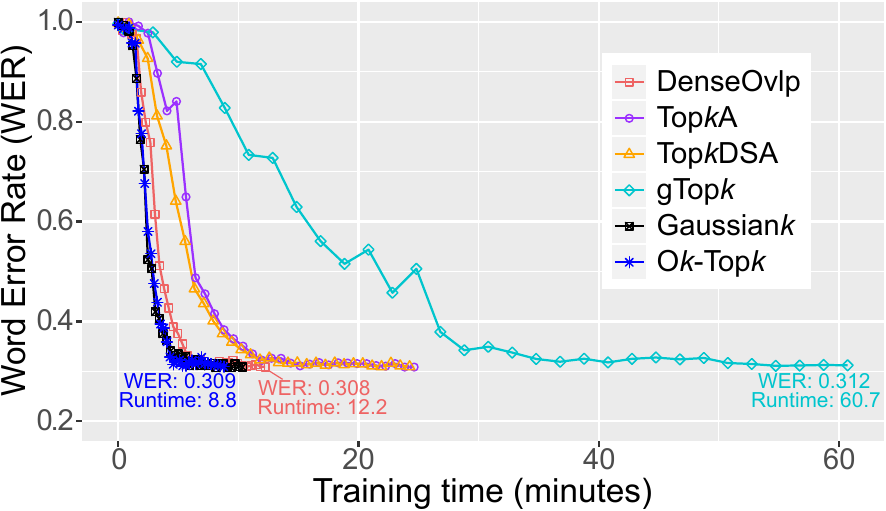}
    \caption{Running on 32 GPU nodes with global batch size = 64.}
    \label{lstm32nodesacc}
  \end{subfigure}
  \begin{subfigure}{\linewidth}
    \centering
  \includegraphics[width=.86\linewidth]{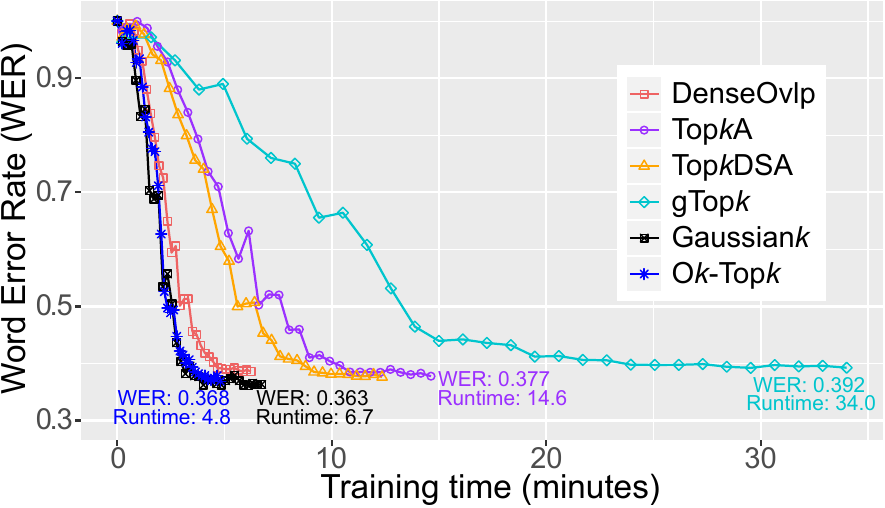}
    \caption{Running on 64 GPU nodes with global batch size = 128.}
    \label{lstm64nodesacc}
  \end{subfigure}
  \caption{WER for LSTM on AN4 with density = 2.0\% training for 160 epochs.}
  \label{lstmacc}
\end{figure}

Figure~\ref{vggacc} presents the Top-1 test accuracy as a function of runtime by training VGG-16 on Cifar-10 for 160 epochs. On both 16 and 32 GPUs, the accuracy achieved by O$k$-Top$k$ is very close to dense allreduce. We did not do any hyperparameter optimization except simply diminishing the learning rate. The accuracy results are consistent with these reported in machine learning community~\cite{ayinde2018building,shi2019understanding}. On both 16 and 32 GPUs, O$k$-Top$k$ achieves the fastest time-to-solution.

\begin{figure*}[ht!]
\centering\includegraphics[angle=90, width=0.99\linewidth]{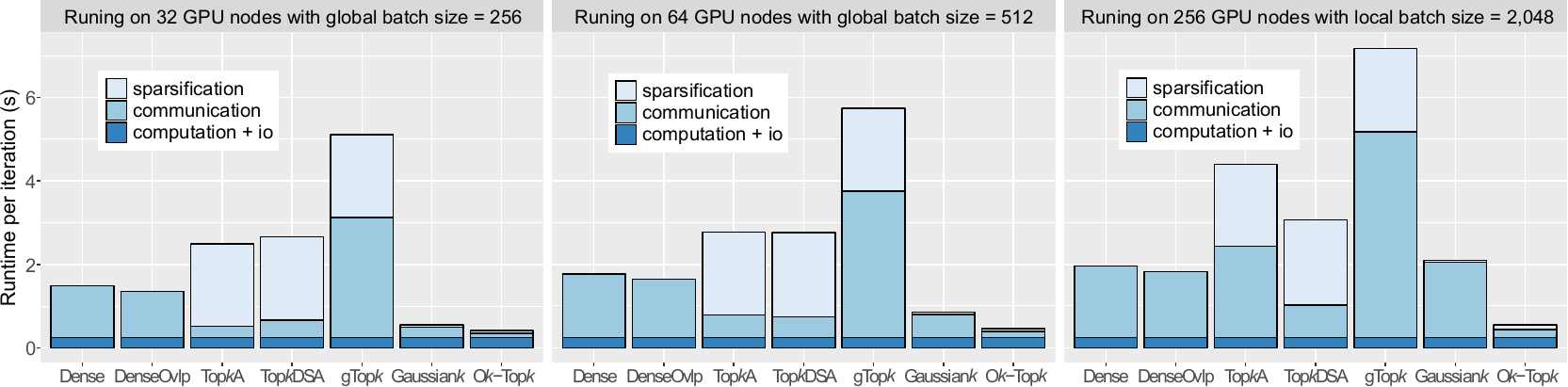}
\caption{\label{berttime} Weak scaling (from 32 to 256 GPU nodes) of training BERT on Wikipedia with density = 1.0\%.}
\end{figure*}

\begin{figure}[ht!]
\centering\includegraphics[width=0.9\linewidth]{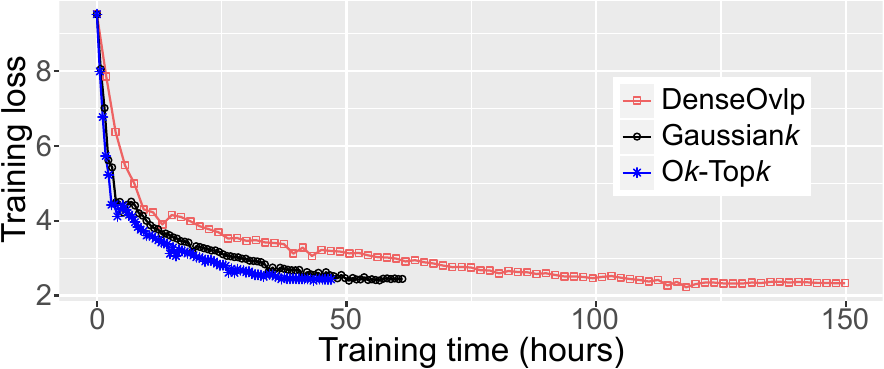}
\caption{\label{bertacc} BERT pre-training on 32 GPU nodes for 400,000 iterations with global batch size = 256 and density = 1.0\%.}
\end{figure}

\subsubsection{Speech recognition}

Figure~\ref{lstmruntime} presents the results of weak scaling for training LSTM on AN4. Similar to the results on VGG-16, O$k$-Top$k$ has a better scalability than the counterparts. On 64 GPUs, O$k$-Top$k$ outperforms the other schemes by 1.34x-7.71x for the total training time.

Figure~\ref{lstmacc} presents the test Word Error Rate (WER, the smaller the better) as a function of runtime by training for 160 epochs. On 32 GPUs, O$k$-Top$k$ is 1.39x faster than DenseOvlp, and achieves 0.309 WER, which is very close to DenseOvlp (0.308). On 64 GPUs, all schemes achieve higher WERs than these on 32 GPUs. This is because the model accuracy is compromised by using a larger global batch size, which is also observed in many other deep learning tasks~\cite{you2018imagenet,you2019largelstm,ben2019demystifying}. How to tune hyperparameters for better accuracy with large batches is not the topic of this work. Surprisingly, on 64 GPUs, O$k$-Top$k$, Gaussian$k$, Top$k$A and Top$k$DSA achieve lower WERs than DenseOvlp, which may be caused by the noise introduced by the sparsification. Overall, on both 32 and 64 GPUs, O$k$-Top$k$ achieves the fastest time-to-solution.

\subsubsection{Natural language processing}

BERT~\cite{devlin2018bert} is a popular language model based on Transformer~\cite{vaswani2017attention}. The model is usually pre-trained on a large dataset and then fine-tuned for various downstream tasks. Pre-training is commonly much more expensive (years on a single GPU) than fine-tuning. Therefore, we focus on pre-training in the evaluation.

Figure~\ref{berttime} presents the results of weak scaling for pre-training BERT. When scaling to 256 GPUs, the communication overhead of Top$k$A and Gaussian$k$ is even higher than the dense allreduce, which again demonstrates that the allgather-based sparse allreduce is not scalable. Top$k$DSA exhibits better scalablity than the allgather-based sparse allreduce, but its communication overhead also significantly increases, since the more workers, the more severe the fill-in problem~\cite{renggli2019sparcml}. On 256 GPUs, O$k$-Top$k$ outperforms all counterparts by 3.29x-12.95x. Using 32 nodes as the baseline, O$k$-Top$k$ achieves 76.3\% parallel efficiency on 256 GPUs in weak scaling, which demonstrates a strong scalalibity of O$k$-Top$k$. 

%wikipedia has 2,500M words
In Figure~\ref{bertacc}, we report the training loss by pre-training BERT from scratch on the Wikipedia dataset (containing 114.5 million sequences with a max length of 128) for 400,000 iterations. Eventually, the training loss of O$k$-Top$k$ decreases to 2.43, which is very close to DenseOvlp (2.33). These results show that O$k$-Top$k$ has a similar convergence rate as the dense allreduce for BERT pre-training. Compared with DenseOvlp, O$k$-Top$k$ reduces the total training time on 32 GPUs from 150 hours to 47 hours (more than 3x speedup), and also outperforms Gaussian$k$ by 1.30x. Since pre-training BERT is very costly (energy- and time-consuming), in Figure~\ref{bertacc} we only present the results for two important baselines (i.e., Gaussian$k$, with the highest training throughput among all baselines, and DenseOvlp, a lossless approach). Since the other baselines are inferior to Gaussian$k$ and DenseOvlp in terms of training throughput and not better than DenseOvlp in terms of convergence rate, it is sufficient to show the advantage of O$k$-Top$k$ by comparing it with these two important baselines in Figure~\ref{bertacc}.

\section{Conclusion}

O$k$-Top$k$ is a novel scheme for distributed deep learning training with sparse gradients. The sparse allreduce of O$k$-Top$k$ incurs less than 6$k$ communication volume, which is  asymptotically optimal and more scalable than the counterparts. O$k$-Top$k$ enables an efficient and accurate top-$k$ values prediction by utilizing the temporal locality of gradient value statistics. Empirical results for data-parallel training of real-world deep learning models on the Piz Daint supercomputer show that O$k$-Top$k$ significantly improves the training throughput while guaranteeing the model accuracy. The throughput improvement would be more significant on commodity clusters with low-bandwidth network. We foresee that our 
scheme will play an important role in scalable distributed training for large-scale models with low communication overhead. In future work, we aim to further utilize O$k$-Top$k$ to reduce the communication overhead in distributed training with a hybrid data and pipeline parallelism~\cite{li2021chimera,fan2021dapple,narayanan2021memory,narayanan2019pipedream}.

\begin{acks}
This project has received funding from the European Research Council (ERC) under the European Union's
Horizon 2020 programme (grant agreement DAPP, No. 678880, EPiGRAM-HS, No. 801039, and MAELSTROM, No. 955513). 
We also thank the Swiss National 
Supercomputing Center for providing the computing resources and technical support.
\end{acks}

\bibliographystyle{ACM-Reference-Format}
\bibliography{mybib}


\begin{thebibliography}{53}

%%% the Web address.

\ifx \showCODEN    \undefined \def \showCODEN     #1{\unskip}     \fi
\ifx \showDOI      \undefined \def \showDOI       #1{#1}\fi
\ifx \showISBNx    \undefined \def \showISBNx     #1{\unskip}     \fi
\ifx \showISBNxiii \undefined \def \showISBNxiii  #1{\unskip}     \fi
\ifx \showISSN     \undefined \def \showISSN      #1{\unskip}     \fi
\ifx \showLCCN     \undefined \def \showLCCN      #1{\unskip}     \fi
\ifx \shownote     \undefined \def \shownote      #1{#1}          \fi
\ifx \showarticletitle \undefined \def \showarticletitle #1{#1}   \fi
\ifx \showURL      \undefined \def \showURL       {\relax}        \fi
% invisible to TeX
\providecommand\bibfield[2]{#2}
\providecommand\bibinfo[2]{#2}
\providecommand\natexlab[1]{#1}
\providecommand\showeprint[2][]{arXiv:#2}

\bibitem[\protect\citeauthoryear{Acero and Stern}{Acero and Stern}{1990}]%
        {acero1990environmental}
\bibfield{author}{\bibinfo{person}{Alejandro Acero} {and}
  \bibinfo{person}{Richard~M Stern}.} \bibinfo{year}{1990}\natexlab{}.
\newblock \showarticletitle{Environmental robustness in automatic speech
  recognition}. In \bibinfo{booktitle}{\emph{International Conference on
  Acoustics, Speech, and Signal Processing}}. IEEE, \bibinfo{pages}{849--852}.
\newblock


\bibitem[\protect\citeauthoryear{Aji and Heafield}{Aji and Heafield}{2017}]%
        {aji2017sparse}
\bibfield{author}{\bibinfo{person}{Alham~Fikri Aji} {and}
  \bibinfo{person}{Kenneth Heafield}.} \bibinfo{year}{2017}\natexlab{}.
\newblock \showarticletitle{Sparse communication for distributed gradient
  descent}.
\newblock \bibinfo{journal}{\emph{arXiv preprint arXiv:1704.05021}}
  (\bibinfo{year}{2017}).
\newblock


\bibitem[\protect\citeauthoryear{Alistarh, Grubic, Li, Tomioka, and
  Vojnovic}{Alistarh et~al\mbox{.}}{2017}]%
        {alistarh2017qsgd}
\bibfield{author}{\bibinfo{person}{Dan Alistarh}, \bibinfo{person}{Demjan
  Grubic}, \bibinfo{person}{Jerry Li}, \bibinfo{person}{Ryota Tomioka}, {and}
  \bibinfo{person}{Milan Vojnovic}.} \bibinfo{year}{2017}\natexlab{}.
\newblock \showarticletitle{QSGD: Communication-efficient SGD via gradient
  quantization and encoding}.
\newblock \bibinfo{journal}{\emph{Advances in Neural Information Processing
  Systems}}  \bibinfo{volume}{30} (\bibinfo{year}{2017}),
  \bibinfo{pages}{1709--1720}.
\newblock


\bibitem[\protect\citeauthoryear{Alistarh, Hoefler, Johansson, Khirirat,
  Konstantinov, and Renggli}{Alistarh et~al\mbox{.}}{2018}]%
        {alistarh2018convergence}
\bibfield{author}{\bibinfo{person}{Dan Alistarh}, \bibinfo{person}{Torsten
  Hoefler}, \bibinfo{person}{Mikael Johansson}, \bibinfo{person}{Sarit
  Khirirat}, \bibinfo{person}{Nikola Konstantinov}, {and}
  \bibinfo{person}{C{\'e}dric Renggli}.} \bibinfo{year}{2018}\natexlab{}.
\newblock \showarticletitle{The convergence of sparsified gradient methods}.
\newblock \bibinfo{journal}{\emph{arXiv preprint arXiv:1809.10505}}
  (\bibinfo{year}{2018}).
\newblock


\bibitem[\protect\citeauthoryear{Alverson, Froese, Kaplan, and Roweth}{Alverson
  et~al\mbox{.}}{2012}]%
        {alverson2012cray}
\bibfield{author}{\bibinfo{person}{Bob Alverson}, \bibinfo{person}{Edwin
  Froese}, \bibinfo{person}{Larry Kaplan}, {and} \bibinfo{person}{Duncan
  Roweth}.} \bibinfo{year}{2012}\natexlab{}.
\newblock \showarticletitle{Cray XC series network}.
\newblock \bibinfo{journal}{\emph{Cray Inc., White Paper WP-Aries01-1112}}
  (\bibinfo{year}{2012}).
\newblock


\bibitem[\protect\citeauthoryear{Ayinde and Zurada}{Ayinde and Zurada}{2018}]%
        {ayinde2018building}
\bibfield{author}{\bibinfo{person}{Babajide~O Ayinde} {and}
  \bibinfo{person}{Jacek~M Zurada}.} \bibinfo{year}{2018}\natexlab{}.
\newblock \showarticletitle{Building efficient convnets using redundant feature
  pruning}.
\newblock \bibinfo{journal}{\emph{arXiv preprint arXiv:1802.07653}}
  (\bibinfo{year}{2018}).
\newblock


\bibitem[\protect\citeauthoryear{Ben-Nun and Hoefler}{Ben-Nun and
  Hoefler}{2019}]%
        {ben2019demystifying}
\bibfield{author}{\bibinfo{person}{Tal Ben-Nun} {and} \bibinfo{person}{Torsten
  Hoefler}.} \bibinfo{year}{2019}\natexlab{}.
\newblock \showarticletitle{Demystifying parallel and distributed deep
  learning: An in-depth concurrency analysis}.
\newblock \bibinfo{journal}{\emph{ACM Computing Surveys (CSUR)}}
  \bibinfo{volume}{52}, \bibinfo{number}{4} (\bibinfo{year}{2019}),
  \bibinfo{pages}{1--43}.
\newblock


\bibitem[\protect\citeauthoryear{Bernstein, Wang, Azizzadenesheli, and
  Anandkumar}{Bernstein et~al\mbox{.}}{2018}]%
        {bernstein2018signsgd}
\bibfield{author}{\bibinfo{person}{Jeremy Bernstein}, \bibinfo{person}{Yu-Xiang
  Wang}, \bibinfo{person}{Kamyar Azizzadenesheli}, {and}
  \bibinfo{person}{Animashree Anandkumar}.} \bibinfo{year}{2018}\natexlab{}.
\newblock \showarticletitle{signSGD: Compressed optimisation for non-convex
  problems}. In \bibinfo{booktitle}{\emph{International Conference on Machine
  Learning}}. PMLR, \bibinfo{pages}{560--569}.
\newblock


\bibitem[\protect\citeauthoryear{Bottou, Curtis, and Nocedal}{Bottou
  et~al\mbox{.}}{2018}]%
        {bottou2018optimization}
\bibfield{author}{\bibinfo{person}{L{\'e}on Bottou}, \bibinfo{person}{Frank~E
  Curtis}, {and} \bibinfo{person}{Jorge Nocedal}.}
  \bibinfo{year}{2018}\natexlab{}.
\newblock \showarticletitle{Optimization methods for large-scale machine
  learning}.
\newblock \bibinfo{journal}{\emph{SIAM Rev.}} \bibinfo{volume}{60},
  \bibinfo{number}{2} (\bibinfo{year}{2018}).
\newblock


\bibitem[\protect\citeauthoryear{Brown, Mann, Ryder, Subbiah, Kaplan, Dhariwal,
  Neelakantan, Shyam, Sastry, Askell, et~al\mbox{.}}{Brown
  et~al\mbox{.}}{2020}]%
        {brown2020language}
\bibfield{author}{\bibinfo{person}{Tom~B Brown}, \bibinfo{person}{Benjamin
  Mann}, \bibinfo{person}{Nick Ryder}, \bibinfo{person}{Melanie Subbiah},
  \bibinfo{person}{Jared Kaplan}, \bibinfo{person}{Prafulla Dhariwal},
  \bibinfo{person}{Arvind Neelakantan}, \bibinfo{person}{Pranav Shyam},
  \bibinfo{person}{Girish Sastry}, \bibinfo{person}{Amanda Askell},
  {et~al\mbox{.}}} \bibinfo{year}{2020}\natexlab{}.
\newblock \showarticletitle{Language models are few-shot learners}.
\newblock \bibinfo{journal}{\emph{arXiv preprint arXiv:2005.14165}}
  (\bibinfo{year}{2020}).
\newblock


\bibitem[\protect\citeauthoryear{Cai, Lin, Xia, Chen, Han, Wang, and Yang}{Cai
  et~al\mbox{.}}{2018}]%
        {cai2018long}
\bibfield{author}{\bibinfo{person}{Yi Cai}, \bibinfo{person}{Yujun Lin},
  \bibinfo{person}{Lixue Xia}, \bibinfo{person}{Xiaoming Chen},
  \bibinfo{person}{Song Han}, \bibinfo{person}{Yu Wang}, {and}
  \bibinfo{person}{Huazhong Yang}.} \bibinfo{year}{2018}\natexlab{}.
\newblock \showarticletitle{Long live time: improving lifetime for
  training-in-memory engines by structured gradient sparsification}. In
  \bibinfo{booktitle}{\emph{Proceedings of the 55th Annual Design Automation
  Conference}}. \bibinfo{pages}{1--6}.
\newblock


\bibitem[\protect\citeauthoryear{Chan, Heimlich, Purkayastha, and Van
  De~Geijn}{Chan et~al\mbox{.}}{2007}]%
        {chan2007collective}
\bibfield{author}{\bibinfo{person}{Ernie Chan}, \bibinfo{person}{Marcel
  Heimlich}, \bibinfo{person}{Avi Purkayastha}, {and} \bibinfo{person}{Robert
  Van De~Geijn}.} \bibinfo{year}{2007}\natexlab{}.
\newblock \showarticletitle{Collective communication: theory, practice, and
  experience}.
\newblock \bibinfo{journal}{\emph{Concurrency and Computation: Practice and
  Experience}} \bibinfo{volume}{19}, \bibinfo{number}{13}
  (\bibinfo{year}{2007}), \bibinfo{pages}{1749--1783}.
\newblock


\bibitem[\protect\citeauthoryear{Devlin, Chang, Lee, and Toutanova}{Devlin
  et~al\mbox{.}}{2018}]%
        {devlin2018bert}
\bibfield{author}{\bibinfo{person}{Jacob Devlin}, \bibinfo{person}{Ming-Wei
  Chang}, \bibinfo{person}{Kenton Lee}, {and} \bibinfo{person}{Kristina
  Toutanova}.} \bibinfo{year}{2018}\natexlab{}.
\newblock \showarticletitle{Bert: Pre-training of deep bidirectional
  transformers for language understanding}.
\newblock \bibinfo{journal}{\emph{arXiv preprint arXiv:1810.04805}}
  (\bibinfo{year}{2018}).
\newblock


\bibitem[\protect\citeauthoryear{Dryden, Moon, Jacobs, and Van~Essen}{Dryden
  et~al\mbox{.}}{2016}]%
        {dryden2016communication}
\bibfield{author}{\bibinfo{person}{Nikoli Dryden}, \bibinfo{person}{Tim Moon},
  \bibinfo{person}{Sam~Ade Jacobs}, {and} \bibinfo{person}{Brian Van~Essen}.}
  \bibinfo{year}{2016}\natexlab{}.
\newblock \showarticletitle{Communication quantization for data-parallel
  training of deep neural networks}. In \bibinfo{booktitle}{\emph{2016 2nd
  Workshop on Machine Learning in HPC Environments (MLHPC)}}. IEEE,
  \bibinfo{pages}{1--8}.
\newblock


\bibitem[\protect\citeauthoryear{Fan, Rong, Meng, Cao, Wang, Zheng, Wu, Long,
  Yang, Xia, et~al\mbox{.}}{Fan et~al\mbox{.}}{2021}]%
        {fan2021dapple}
\bibfield{author}{\bibinfo{person}{Shiqing Fan}, \bibinfo{person}{Yi Rong},
  \bibinfo{person}{Chen Meng}, \bibinfo{person}{Zongyan Cao},
  \bibinfo{person}{Siyu Wang}, \bibinfo{person}{Zhen Zheng},
  \bibinfo{person}{Chuan Wu}, \bibinfo{person}{Guoping Long},
  \bibinfo{person}{Jun Yang}, \bibinfo{person}{Lixue Xia}, {et~al\mbox{.}}}
  \bibinfo{year}{2021}\natexlab{}.
\newblock \showarticletitle{DAPPLE: a pipelined data parallel approach for
  training large models}. In \bibinfo{booktitle}{\emph{Proceedings of the 26th
  ACM SIGPLAN Symposium on Principles and Practice of Parallel Programming}}.
  \bibinfo{pages}{431--445}.
\newblock


\bibitem[\protect\citeauthoryear{Fei, Ho, Sahu, Canini, and Sapio}{Fei
  et~al\mbox{.}}{2021}]%
        {fei2021efficient}
\bibfield{author}{\bibinfo{person}{Jiawei Fei}, \bibinfo{person}{Chen-Yu Ho},
  \bibinfo{person}{Atal~N Sahu}, \bibinfo{person}{Marco Canini}, {and}
  \bibinfo{person}{Amedeo Sapio}.} \bibinfo{year}{2021}\natexlab{}.
\newblock \showarticletitle{Efficient sparse collective communication and its
  application to accelerate distributed deep learning}. In
  \bibinfo{booktitle}{\emph{Proceedings of the 2021 ACM SIGCOMM 2021
  Conference}}. \bibinfo{pages}{676--691}.
\newblock


\bibitem[\protect\citeauthoryear{Foley and Danskin}{Foley and Danskin}{2017}]%
        {foley2017ultra}
\bibfield{author}{\bibinfo{person}{Denis Foley} {and} \bibinfo{person}{John
  Danskin}.} \bibinfo{year}{2017}\natexlab{}.
\newblock \showarticletitle{{Ultra-performance Pascal GPU and NVLink
  Interconnect}}.
\newblock \bibinfo{journal}{\emph{IEEE Micro}} \bibinfo{volume}{37},
  \bibinfo{number}{2} (\bibinfo{year}{2017}), \bibinfo{pages}{7--17}.
\newblock


\bibitem[\protect\citeauthoryear{Goyal, Doll{\'a}r, Girshick, Noordhuis,
  Wesolowski, Kyrola, Tulloch, Jia, and He}{Goyal et~al\mbox{.}}{2017}]%
        {goyal2017accurate}
\bibfield{author}{\bibinfo{person}{Priya Goyal}, \bibinfo{person}{Piotr
  Doll{\'a}r}, \bibinfo{person}{Ross Girshick}, \bibinfo{person}{Pieter
  Noordhuis}, \bibinfo{person}{Lukasz Wesolowski}, \bibinfo{person}{Aapo
  Kyrola}, \bibinfo{person}{Andrew Tulloch}, \bibinfo{person}{Yangqing Jia},
  {and} \bibinfo{person}{Kaiming He}.} \bibinfo{year}{2017}\natexlab{}.
\newblock \showarticletitle{Accurate, large minibatch sgd: Training imagenet in
  1 hour}.
\newblock \bibinfo{journal}{\emph{arXiv preprint arXiv:1706.02677}}
  (\bibinfo{year}{2017}).
\newblock


\bibitem[\protect\citeauthoryear{Han, Wang, and Leung}{Han
  et~al\mbox{.}}{2020}]%
        {han2020adaptive}
\bibfield{author}{\bibinfo{person}{Pengchao Han}, \bibinfo{person}{Shiqiang
  Wang}, {and} \bibinfo{person}{Kin~K Leung}.} \bibinfo{year}{2020}\natexlab{}.
\newblock \showarticletitle{Adaptive gradient sparsification for efficient
  federated learning: An online learning approach}. In
  \bibinfo{booktitle}{\emph{2020 IEEE 40th International Conference on
  Distributed Computing Systems (ICDCS)}}. IEEE, \bibinfo{pages}{300--310}.
\newblock


\bibitem[\protect\citeauthoryear{Hensgen, Finkel, and Manber}{Hensgen
  et~al\mbox{.}}{1988}]%
        {hensgen1988two}
\bibfield{author}{\bibinfo{person}{Debra Hensgen}, \bibinfo{person}{Raphael
  Finkel}, {and} \bibinfo{person}{Udi Manber}.}
  \bibinfo{year}{1988}\natexlab{}.
\newblock \showarticletitle{Two algorithms for barrier synchronization}.
\newblock \bibinfo{journal}{\emph{International Journal of Parallel
  Programming}} \bibinfo{volume}{17}, \bibinfo{number}{1}
  (\bibinfo{year}{1988}), \bibinfo{pages}{1--17}.
\newblock


\bibitem[\protect\citeauthoryear{Hochreiter and Schmidhuber}{Hochreiter and
  Schmidhuber}{1997}]%
        {hochreiter1997long}
\bibfield{author}{\bibinfo{person}{Sepp Hochreiter} {and}
  \bibinfo{person}{J{\"u}rgen Schmidhuber}.} \bibinfo{year}{1997}\natexlab{}.
\newblock \showarticletitle{Long short-term memory}.
\newblock \bibinfo{journal}{\emph{Neural computation}} \bibinfo{volume}{9},
  \bibinfo{number}{8} (\bibinfo{year}{1997}), \bibinfo{pages}{1735--1780}.
\newblock


\bibitem[\protect\citeauthoryear{Hoefler, Alistarh, Ben-Nun, Dryden, and
  Peste}{Hoefler et~al\mbox{.}}{2021}]%
        {hoefler2021sparsity}
\bibfield{author}{\bibinfo{person}{Torsten Hoefler}, \bibinfo{person}{Dan
  Alistarh}, \bibinfo{person}{Tal Ben-Nun}, \bibinfo{person}{Nikoli Dryden},
  {and} \bibinfo{person}{Alexandra Peste}.} \bibinfo{year}{2021}\natexlab{}.
\newblock \showarticletitle{Sparsity in Deep Learning: Pruning and growth for
  efficient inference and training in neural networks}.
\newblock \bibinfo{journal}{\emph{arXiv preprint arXiv:2102.00554}}
  (\bibinfo{year}{2021}).
\newblock


\bibitem[\protect\citeauthoryear{Horv{\'a}th, Kovalev, Mishchenko, Stich, and
  Richt{\'a}rik}{Horv{\'a}th et~al\mbox{.}}{2019}]%
        {horvath2019stochastic}
\bibfield{author}{\bibinfo{person}{Samuel Horv{\'a}th}, \bibinfo{person}{Dmitry
  Kovalev}, \bibinfo{person}{Konstantin Mishchenko}, \bibinfo{person}{Sebastian
  Stich}, {and} \bibinfo{person}{Peter Richt{\'a}rik}.}
  \bibinfo{year}{2019}\natexlab{}.
\newblock \showarticletitle{Stochastic distributed learning with gradient
  quantization and variance reduction}.
\newblock \bibinfo{journal}{\emph{arXiv preprint arXiv:1904.05115}}
  (\bibinfo{year}{2019}).
\newblock


\bibitem[\protect\citeauthoryear{Kingma and Ba}{Kingma and Ba}{2014}]%
        {kingma2014adam}
\bibfield{author}{\bibinfo{person}{Diederik~P Kingma} {and}
  \bibinfo{person}{Jimmy Ba}.} \bibinfo{year}{2014}\natexlab{}.
\newblock \showarticletitle{Adam: A method for stochastic optimization}.
\newblock \bibinfo{journal}{\emph{arXiv preprint arXiv:1412.6980}}
  (\bibinfo{year}{2014}).
\newblock


\bibitem[\protect\citeauthoryear{Li, Ben-Nun, Girolamo, Alistarh, and
  Hoefler}{Li et~al\mbox{.}}{2020a}]%
        {li2020taming}
\bibfield{author}{\bibinfo{person}{Shigang Li}, \bibinfo{person}{Tal Ben-Nun},
  \bibinfo{person}{Salvatore~Di Girolamo}, \bibinfo{person}{Dan Alistarh},
  {and} \bibinfo{person}{Torsten Hoefler}.} \bibinfo{year}{2020}\natexlab{a}.
\newblock \showarticletitle{Taming unbalanced training workloads in deep
  learning with partial collective operations}. In
  \bibinfo{booktitle}{\emph{Proceedings of the 25th ACM SIGPLAN Symposium on
  Principles and Practice of Parallel Programming}}.
\newblock


\bibitem[\protect\citeauthoryear{Li, Ben-Nun, Nadiradze, Di~Girolamo, Dryden,
  Alistarh, and Hoefler}{Li et~al\mbox{.}}{2020b}]%
        {li2020breaking}
\bibfield{author}{\bibinfo{person}{Shigang Li}, \bibinfo{person}{Tal Ben-Nun},
  \bibinfo{person}{Giorgi Nadiradze}, \bibinfo{person}{Salvatore Di~Girolamo},
  \bibinfo{person}{Nikoli Dryden}, \bibinfo{person}{Dan Alistarh}, {and}
  \bibinfo{person}{Torsten Hoefler}.} \bibinfo{year}{2020}\natexlab{b}.
\newblock \showarticletitle{Breaking (global) barriers in parallel stochastic
  optimization with wait-avoiding group averaging}.
\newblock \bibinfo{journal}{\emph{IEEE Transactions on Parallel and Distributed
  Systems}} \bibinfo{volume}{32}, \bibinfo{number}{7} (\bibinfo{year}{2020}),
  \bibinfo{pages}{1725--1739}.
\newblock


\bibitem[\protect\citeauthoryear{Li and Hoefler}{Li and Hoefler}{2021}]%
        {li2021chimera}
\bibfield{author}{\bibinfo{person}{Shigang Li} {and} \bibinfo{person}{Torsten
  Hoefler}.} \bibinfo{year}{2021}\natexlab{}.
\newblock \showarticletitle{Chimera: Efficiently Training Large-Scale Neural
  Networks with Bidirectional Pipelines}.
\newblock \bibinfo{journal}{\emph{arXiv preprint arXiv:2107.06925}}
  (\bibinfo{year}{2021}).
\newblock


\bibitem[\protect\citeauthoryear{Li, Hoefler, and Snir}{Li
  et~al\mbox{.}}{2013}]%
        {li2013numa}
\bibfield{author}{\bibinfo{person}{Shigang Li}, \bibinfo{person}{Torsten
  Hoefler}, {and} \bibinfo{person}{Marc Snir}.}
  \bibinfo{year}{2013}\natexlab{}.
\newblock \showarticletitle{NUMA-aware shared-memory collective communication
  for MPI}. In \bibinfo{booktitle}{\emph{Proceedings of the 22nd international
  symposium on High-performance parallel and distributed computing}}.
  \bibinfo{pages}{85--96}.
\newblock


\bibitem[\protect\citeauthoryear{Mahmoud, Modarres, and Smythe}{Mahmoud
  et~al\mbox{.}}{1995}]%
        {mahmoud1995analysis}
\bibfield{author}{\bibinfo{person}{Hosam~M Mahmoud}, \bibinfo{person}{Reza
  Modarres}, {and} \bibinfo{person}{Robert~T Smythe}.}
  \bibinfo{year}{1995}\natexlab{}.
\newblock \showarticletitle{Analysis of quickselect: An algorithm for order
  statistics}.
\newblock \bibinfo{journal}{\emph{RAIRO-Theoretical Informatics and
  Applications-Informatique Th{\'e}orique et Applications}}
  \bibinfo{volume}{29}, \bibinfo{number}{4} (\bibinfo{year}{1995}),
  \bibinfo{pages}{255--276}.
\newblock


\bibitem[\protect\citeauthoryear{Nadiradze, Sabour, Davies, Li, and
  Alistarh}{Nadiradze et~al\mbox{.}}{2021}]%
        {nadiradze2021asynchronous}
\bibfield{author}{\bibinfo{person}{Giorgi Nadiradze},
  \bibinfo{person}{Amirmojtaba Sabour}, \bibinfo{person}{Peter Davies},
  \bibinfo{person}{Shigang Li}, {and} \bibinfo{person}{Dan Alistarh}.}
  \bibinfo{year}{2021}\natexlab{}.
\newblock \showarticletitle{Asynchronous decentralized SGD with quantized and
  local updates}.
\newblock \bibinfo{journal}{\emph{Advances in Neural Information Processing
  Systems}}  \bibinfo{volume}{34} (\bibinfo{year}{2021}).
\newblock


\bibitem[\protect\citeauthoryear{Narayanan, Harlap, Phanishayee, Seshadri,
  Devanur, Ganger, Gibbons, and Zaharia}{Narayanan et~al\mbox{.}}{2019}]%
        {narayanan2019pipedream}
\bibfield{author}{\bibinfo{person}{Deepak Narayanan}, \bibinfo{person}{Aaron
  Harlap}, \bibinfo{person}{Amar Phanishayee}, \bibinfo{person}{Vivek
  Seshadri}, \bibinfo{person}{Nikhil~R Devanur}, \bibinfo{person}{Gregory~R
  Ganger}, \bibinfo{person}{Phillip~B Gibbons}, {and} \bibinfo{person}{Matei
  Zaharia}.} \bibinfo{year}{2019}\natexlab{}.
\newblock \showarticletitle{PipeDream: generalized pipeline parallelism for DNN
  training}. In \bibinfo{booktitle}{\emph{Proceedings of the 27th ACM Symposium
  on Operating Systems Principles}}. \bibinfo{pages}{1--15}.
\newblock


\bibitem[\protect\citeauthoryear{Narayanan, Phanishayee, Shi, Chen, and
  Zaharia}{Narayanan et~al\mbox{.}}{2021}]%
        {narayanan2021memory}
\bibfield{author}{\bibinfo{person}{Deepak Narayanan}, \bibinfo{person}{Amar
  Phanishayee}, \bibinfo{person}{Kaiyu Shi}, \bibinfo{person}{Xie Chen}, {and}
  \bibinfo{person}{Matei Zaharia}.} \bibinfo{year}{2021}\natexlab{}.
\newblock \showarticletitle{Memory-efficient pipeline-parallel dnn training}.
  In \bibinfo{booktitle}{\emph{International Conference on Machine Learning}}.
  PMLR, \bibinfo{pages}{7937--7947}.
\newblock


\bibitem[\protect\citeauthoryear{Paszke, Gross, Massa, Lerer, Bradbury, Chanan,
  Killeen, Lin, Gimelshein, Antiga, et~al\mbox{.}}{Paszke
  et~al\mbox{.}}{2019}]%
        {paszke2019pytorch}
\bibfield{author}{\bibinfo{person}{Adam Paszke}, \bibinfo{person}{Sam Gross},
  \bibinfo{person}{Francisco Massa}, \bibinfo{person}{Adam Lerer},
  \bibinfo{person}{James Bradbury}, \bibinfo{person}{Gregory Chanan},
  \bibinfo{person}{Trevor Killeen}, \bibinfo{person}{Zeming Lin},
  \bibinfo{person}{Natalia Gimelshein}, \bibinfo{person}{Luca Antiga},
  {et~al\mbox{.}}} \bibinfo{year}{2019}\natexlab{}.
\newblock \showarticletitle{Pytorch: An imperative style, high-performance deep
  learning library}. In \bibinfo{booktitle}{\emph{Advances in neural
  information processing systems}}. \bibinfo{pages}{8026--8037}.
\newblock


\bibitem[\protect\citeauthoryear{Radford, Wu, Child, Luan, Amodei, and
  Sutskever}{Radford et~al\mbox{.}}{2019}]%
        {radford2019language}
\bibfield{author}{\bibinfo{person}{Alec Radford}, \bibinfo{person}{Jeffrey Wu},
  \bibinfo{person}{Rewon Child}, \bibinfo{person}{David Luan},
  \bibinfo{person}{Dario Amodei}, {and} \bibinfo{person}{Ilya Sutskever}.}
  \bibinfo{year}{2019}\natexlab{}.
\newblock \showarticletitle{Language models are unsupervised multitask
  learners}.
\newblock \bibinfo{journal}{\emph{OpenAI blog}} \bibinfo{volume}{1},
  \bibinfo{number}{8} (\bibinfo{year}{2019}), \bibinfo{pages}{9}.
\newblock


\bibitem[\protect\citeauthoryear{Real, Aggarwal, Huang, and Le}{Real
  et~al\mbox{.}}{2019}]%
        {real2019regularized}
\bibfield{author}{\bibinfo{person}{Esteban Real}, \bibinfo{person}{Alok
  Aggarwal}, \bibinfo{person}{Yanping Huang}, {and} \bibinfo{person}{Quoc~V
  Le}.} \bibinfo{year}{2019}\natexlab{}.
\newblock \showarticletitle{Regularized evolution for image classifier
  architecture search}. In \bibinfo{booktitle}{\emph{Proceedings of the aaai
  conference on artificial intelligence}}, Vol.~\bibinfo{volume}{33}.
  \bibinfo{pages}{4780--4789}.
\newblock


\bibitem[\protect\citeauthoryear{Renggli, Ashkboos, Aghagolzadeh, Alistarh, and
  Hoefler}{Renggli et~al\mbox{.}}{2019}]%
        {renggli2019sparcml}
\bibfield{author}{\bibinfo{person}{C{\`e}dric Renggli}, \bibinfo{person}{Saleh
  Ashkboos}, \bibinfo{person}{Mehdi Aghagolzadeh}, \bibinfo{person}{Dan
  Alistarh}, {and} \bibinfo{person}{Torsten Hoefler}.}
  \bibinfo{year}{2019}\natexlab{}.
\newblock \showarticletitle{SparCML: High-performance sparse communication for
  machine learning}. In \bibinfo{booktitle}{\emph{Proceedings of the
  International Conference for High Performance Computing, Networking, Storage
  and Analysis}}. \bibinfo{pages}{1--15}.
\newblock


\bibitem[\protect\citeauthoryear{Sensi, Girolamo, McMahon, Roweth, and
  Hoefler}{Sensi et~al\mbox{.}}{2020}]%
        {slingshot}
\bibfield{author}{\bibinfo{person}{Daniele~De Sensi},
  \bibinfo{person}{Salvatore~Di Girolamo}, \bibinfo{person}{Kim~H. McMahon},
  \bibinfo{person}{Duncan Roweth}, {and} \bibinfo{person}{Torsten Hoefler}.}
  \bibinfo{year}{2020}\natexlab{}.
\newblock \showarticletitle{{An In-Depth Analysis of the Slingshot
  Interconnect}}. In \bibinfo{booktitle}{\emph{Proceedings of the International
  Conference for High Performance Computing, Networking, Storage and Analysis
  (SC20)}}.
\newblock


\bibitem[\protect\citeauthoryear{Sergeev and Del~Balso}{Sergeev and
  Del~Balso}{2018}]%
        {sergeev2018horovod}
\bibfield{author}{\bibinfo{person}{Alexander Sergeev} {and}
  \bibinfo{person}{Mike Del~Balso}.} \bibinfo{year}{2018}\natexlab{}.
\newblock \showarticletitle{Horovod: fast and easy distributed deep learning in
  {TensorFlow}}.
\newblock \bibinfo{journal}{\emph{arXiv preprint arXiv:1802.05799}}
  (\bibinfo{year}{2018}).
\newblock


\bibitem[\protect\citeauthoryear{Shanbhag, Pirk, and Madden}{Shanbhag
  et~al\mbox{.}}{2018}]%
        {shanbhag2018efficient}
\bibfield{author}{\bibinfo{person}{Anil Shanbhag}, \bibinfo{person}{Holger
  Pirk}, {and} \bibinfo{person}{Samuel Madden}.}
  \bibinfo{year}{2018}\natexlab{}.
\newblock \showarticletitle{Efficient top-k query processing on massively
  parallel hardware}. In \bibinfo{booktitle}{\emph{Proceedings of the 2018
  International Conference on Management of Data}}.
  \bibinfo{pages}{1557--1570}.
\newblock


\bibitem[\protect\citeauthoryear{Shanley}{Shanley}{2003}]%
        {shanley2003infiniband}
\bibfield{author}{\bibinfo{person}{Tom Shanley}.}
  \bibinfo{year}{2003}\natexlab{}.
\newblock \bibinfo{booktitle}{\emph{InfiniBand network architecture}}.
\newblock \bibinfo{publisher}{Addison-Wesley Professional}.
\newblock


\bibitem[\protect\citeauthoryear{Shi, Chu, Cheung, and See}{Shi
  et~al\mbox{.}}{2019a}]%
        {shi2019understanding}
\bibfield{author}{\bibinfo{person}{Shaohuai Shi}, \bibinfo{person}{Xiaowen
  Chu}, \bibinfo{person}{Ka~Chun Cheung}, {and} \bibinfo{person}{Simon See}.}
  \bibinfo{year}{2019}\natexlab{a}.
\newblock \showarticletitle{Understanding top-k sparsification in distributed
  deep learning}.
\newblock \bibinfo{journal}{\emph{arXiv preprint arXiv:1911.08772}}
  (\bibinfo{year}{2019}).
\newblock


\bibitem[\protect\citeauthoryear{Shi, Wang, Zhao, Tang, Wang, Huang, and
  Chu}{Shi et~al\mbox{.}}{2019b}]%
        {shi2019distributed}
\bibfield{author}{\bibinfo{person}{Shaohuai Shi}, \bibinfo{person}{Qiang Wang},
  \bibinfo{person}{Kaiyong Zhao}, \bibinfo{person}{Zhenheng Tang},
  \bibinfo{person}{Yuxin Wang}, \bibinfo{person}{Xiang Huang}, {and}
  \bibinfo{person}{Xiaowen Chu}.} \bibinfo{year}{2019}\natexlab{b}.
\newblock \showarticletitle{A distributed synchronous SGD algorithm with global
  top-k sparsification for low bandwidth networks}. In
  \bibinfo{booktitle}{\emph{2019 IEEE 39th International Conference on
  Distributed Computing Systems (ICDCS)}}. IEEE, \bibinfo{pages}{2238--2247}.
\newblock


\bibitem[\protect\citeauthoryear{Shi, Zhao, Wang, Tang, and Chu}{Shi
  et~al\mbox{.}}{2019c}]%
        {shi2019convergence}
\bibfield{author}{\bibinfo{person}{Shaohuai Shi}, \bibinfo{person}{Kaiyong
  Zhao}, \bibinfo{person}{Qiang Wang}, \bibinfo{person}{Zhenheng Tang}, {and}
  \bibinfo{person}{Xiaowen Chu}.} \bibinfo{year}{2019}\natexlab{c}.
\newblock \showarticletitle{A Convergence Analysis of Distributed SGD with
  Communication-Efficient Gradient Sparsification.}. In
  \bibinfo{booktitle}{\emph{IJCAI}}. \bibinfo{pages}{3411--3417}.
\newblock


\bibitem[\protect\citeauthoryear{Simonyan and Zisserman}{Simonyan and
  Zisserman}{2014}]%
        {simonyan2014very}
\bibfield{author}{\bibinfo{person}{Karen Simonyan} {and}
  \bibinfo{person}{Andrew Zisserman}.} \bibinfo{year}{2014}\natexlab{}.
\newblock \showarticletitle{Very deep convolutional networks for large-scale
  image recognition}.
\newblock \bibinfo{journal}{\emph{arXiv preprint arXiv:1409.1556}}
  (\bibinfo{year}{2014}).
\newblock


\bibitem[\protect\citeauthoryear{Thakur, Rabenseifner, and Gropp}{Thakur
  et~al\mbox{.}}{2005}]%
        {thakur2005optimization}
\bibfield{author}{\bibinfo{person}{Rajeev Thakur}, \bibinfo{person}{Rolf
  Rabenseifner}, {and} \bibinfo{person}{William Gropp}.}
  \bibinfo{year}{2005}\natexlab{}.
\newblock \showarticletitle{Optimization of collective communication operations
  in MPICH}.
\newblock \bibinfo{journal}{\emph{The International Journal of High Performance
  Computing Applications}} \bibinfo{volume}{19}, \bibinfo{number}{1}
  (\bibinfo{year}{2005}), \bibinfo{pages}{49--66}.
\newblock


\bibitem[\protect\citeauthoryear{Vaswani, Shazeer, Parmar, Uszkoreit, Jones,
  Gomez, Kaiser, and Polosukhin}{Vaswani et~al\mbox{.}}{2017}]%
        {vaswani2017attention}
\bibfield{author}{\bibinfo{person}{Ashish Vaswani}, \bibinfo{person}{Noam
  Shazeer}, \bibinfo{person}{Niki Parmar}, \bibinfo{person}{Jakob Uszkoreit},
  \bibinfo{person}{Llion Jones}, \bibinfo{person}{Aidan~N Gomez},
  \bibinfo{person}{Lukasz Kaiser}, {and} \bibinfo{person}{Illia Polosukhin}.}
  \bibinfo{year}{2017}\natexlab{}.
\newblock \showarticletitle{Attention is all you need}.
\newblock \bibinfo{journal}{\emph{arXiv preprint arXiv:1706.03762}}
  (\bibinfo{year}{2017}).
\newblock


\bibitem[\protect\citeauthoryear{Wang, Wu, Zhang, Liu, Bosilca, Herlihy, and
  Fonseca}{Wang et~al\mbox{.}}{2020}]%
        {wang2020fft}
\bibfield{author}{\bibinfo{person}{Linnan Wang}, \bibinfo{person}{Wei Wu},
  \bibinfo{person}{Junyu Zhang}, \bibinfo{person}{Hang Liu},
  \bibinfo{person}{George Bosilca}, \bibinfo{person}{Maurice Herlihy}, {and}
  \bibinfo{person}{Rodrigo Fonseca}.} \bibinfo{year}{2020}\natexlab{}.
\newblock \showarticletitle{FFT-based Gradient Sparsification for the
  Distributed Training of Deep Neural Networks}. In
  \bibinfo{booktitle}{\emph{Proceedings of the 29th International Symposium on
  High-Performance Parallel and Distributed Computing}}.
  \bibinfo{pages}{113--124}.
\newblock


\bibitem[\protect\citeauthoryear{Wen, Xu, Yan, Wu, Wang, Chen, and Li}{Wen
  et~al\mbox{.}}{2017}]%
        {10.5555/3294771.3294915}
\bibfield{author}{\bibinfo{person}{Wei Wen}, \bibinfo{person}{Cong Xu},
  \bibinfo{person}{Feng Yan}, \bibinfo{person}{Chunpeng Wu},
  \bibinfo{person}{Yandan Wang}, \bibinfo{person}{Yiran Chen}, {and}
  \bibinfo{person}{Hai Li}.} \bibinfo{year}{2017}\natexlab{}.
\newblock \showarticletitle{TernGrad: Ternary Gradients to Reduce Communication
  in Distributed Deep Learning}. In \bibinfo{booktitle}{\emph{Proceedings of
  the 31st International Conference on Neural Information Processing Systems}}
  (Long Beach, California, USA) \emph{(\bibinfo{series}{NIPS'17})}.
  \bibinfo{publisher}{Curran Associates Inc.}, \bibinfo{address}{Red Hook, NY,
  USA}, \bibinfo{pages}{1508–1518}.
\newblock


\bibitem[\protect\citeauthoryear{Wu, Dong, Li, Huang, and Dai}{Wu
  et~al\mbox{.}}{2019}]%
        {wu2019network}
\bibfield{author}{\bibinfo{person}{Ke Wu}, \bibinfo{person}{Dezun Dong},
  \bibinfo{person}{Cunlu Li}, \bibinfo{person}{Shan Huang}, {and}
  \bibinfo{person}{Yi Dai}.} \bibinfo{year}{2019}\natexlab{}.
\newblock \showarticletitle{Network congestion avoidance through
  packet-chaining reservation}. In \bibinfo{booktitle}{\emph{Proceedings of the
  48th International Conference on Parallel Processing}}.
  \bibinfo{pages}{1--10}.
\newblock


\bibitem[\protect\citeauthoryear{Xu, Kostopoulou, Dutta, Li, Ntoulas, and
  Kalnis}{Xu et~al\mbox{.}}{2021}]%
        {xu2021deepreduce}
\bibfield{author}{\bibinfo{person}{Hang Xu}, \bibinfo{person}{Kelly
  Kostopoulou}, \bibinfo{person}{Aritra Dutta}, \bibinfo{person}{Xin Li},
  \bibinfo{person}{Alexandros Ntoulas}, {and} \bibinfo{person}{Panos Kalnis}.}
  \bibinfo{year}{2021}\natexlab{}.
\newblock \showarticletitle{DeepReduce: A Sparse-tensor Communication Framework
  for Federated Deep Learning}.
\newblock \bibinfo{journal}{\emph{Advances in Neural Information Processing
  Systems}}  \bibinfo{volume}{34} (\bibinfo{year}{2021}).
\newblock


\bibitem[\protect\citeauthoryear{You, Hseu, Ying, Demmel, Keutzer, and
  Hsieh}{You et~al\mbox{.}}{2019a}]%
        {you2019largelstm}
\bibfield{author}{\bibinfo{person}{Yang You}, \bibinfo{person}{Jonathan Hseu},
  \bibinfo{person}{Chris Ying}, \bibinfo{person}{James Demmel},
  \bibinfo{person}{Kurt Keutzer}, {and} \bibinfo{person}{Cho-Jui Hsieh}.}
  \bibinfo{year}{2019}\natexlab{a}.
\newblock \showarticletitle{{Large-batch training for LSTM and beyond}}. In
  \bibinfo{booktitle}{\emph{Proceedings of the International Conference for
  High Performance Computing, Networking, Storage and Analysis}}.
  \bibinfo{pages}{1--16}.
\newblock


\bibitem[\protect\citeauthoryear{You, Li, Reddi, Hseu, Kumar, Bhojanapalli,
  Song, Demmel, Keutzer, and Hsieh}{You et~al\mbox{.}}{2019b}]%
        {you2019largebert}
\bibfield{author}{\bibinfo{person}{Yang You}, \bibinfo{person}{Jing Li},
  \bibinfo{person}{Sashank Reddi}, \bibinfo{person}{Jonathan Hseu},
  \bibinfo{person}{Sanjiv Kumar}, \bibinfo{person}{Srinadh Bhojanapalli},
  \bibinfo{person}{Xiaodan Song}, \bibinfo{person}{James Demmel},
  \bibinfo{person}{Kurt Keutzer}, {and} \bibinfo{person}{Cho-Jui Hsieh}.}
  \bibinfo{year}{2019}\natexlab{b}.
\newblock \showarticletitle{Large batch optimization for deep learning:
  Training bert in 76 minutes}.
\newblock \bibinfo{journal}{\emph{arXiv preprint arXiv:1904.00962}}
  (\bibinfo{year}{2019}).
\newblock


\bibitem[\protect\citeauthoryear{You, Zhang, Hsieh, Demmel, and Keutzer}{You
  et~al\mbox{.}}{2018}]%
        {you2018imagenet}
\bibfield{author}{\bibinfo{person}{Yang You}, \bibinfo{person}{Zhao Zhang},
  \bibinfo{person}{Cho-Jui Hsieh}, \bibinfo{person}{James Demmel}, {and}
  \bibinfo{person}{Kurt Keutzer}.} \bibinfo{year}{2018}\natexlab{}.
\newblock \showarticletitle{{Imagenet training in minutes}}. In
  \bibinfo{booktitle}{\emph{Proceedings of the 47th International Conference on
  Parallel Processing}}. \bibinfo{pages}{1--10}.
\newblock


\end{thebibliography}

\clearpage

\appendix
\section{Artifact Appendix}

%%%%%%%%%%%%%%%%%%%%%%%%%%%%%%%%%%%%%%%%%%%%%%%%%%%%%%%%%%%%%%%%%%%%%
\subsection{Abstract}

The artifact contains the source code for our O\textit{k}-Top\textit{k} and the benchmarks used in the evaluation. It supports the results in Section~\ref{sec:eval}. To validate or reproduce the results, build this artifact and check the results returned by running benchmarks. 

\subsection{Artifact check-list (meta-information)}

{\small
\begin{itemize}
  \item {\bf Algorithm:} O\textit{k}-Top\textit{k}
  \item {\bf Compilation:} Python 3.8
  \item {\bf Model:} VGG, LSTM, BERT
  \item {\bf Data set:} Cifar-10, AN4, Wikipedia
  \item {\bf Run-time environment:} torch, mpi4py, apex
  \item {\bf Hardware:} GPU clusters
  \item {\bf Execution:} srun or mpirun
  \item {\bf Metrics:} execution time, top-1 test accuracy, WER, training loss
  \item {\bf Output:} txt files
  \item {\bf How much disk space required (approximately)?:} \\ 100 GB
  \item {\bf How much time is needed to prepare workflow (approximately)?:} Except for preparing the Wikipedia dataset, it takes about 30 minutes. Preprocessing the Wikipedia dataset takes several hours.
  \item {\bf How much time is needed to complete experiments (approximately)?:} It takes about 5 hours without pre-training BERT. BERT pre-training takes more than 100 hours.
  \item {\bf Archived?:} Yes. \url{https://doi.org/10.5281/zenodo.5808267}
\end{itemize}

%%%%%%%%%%%%%%%%%%%%%%%%%%%%%%%%%%%%%%%%%%%%%%%%%%%%%%%%%%%%%%%%%%%%%
\subsection{Description}

\subsubsection{How to access}

The artifact can be downloaded from \\ \url{https://github.com/Shigangli/Ok-Topk}

The artifact can also be downloaded using the DOI link \\ \url{https://doi.org/10.5281/zenodo.5808267}

\subsubsection{Hardware dependencies}
GPU clusters

\subsubsection{Software dependencies}
To run the experiments, Python 3.8 is required. Python packages, including \texttt{torch, mpi4py, apex, simplejson, tensorboard, tensorboardX, ujson, tqdm, h5py, coloredlogs, psutil, torchaudio, torchvision, numba, librosa, python-Levenshtein}, and \texttt{warpctc-pytorch}, are required.

\subsection{Installation}

1. Download the artifact and extract it in your personal \texttt{\$WORK} directory.
\\

\noindent 2. Setup Python environment and install the dependent packages.

\texttt{> conda create --name py38\_oktopk python=3.8}

\texttt{> conda activate py38\_oktopk}
\\

\texttt{> pip3 install pip==20.2.4}

\texttt{> pip install -r requirements.txt}

\texttt{> MPICC="cc -shared" pip install --no-binary=mpi4py mpi4py}
\\

\texttt{> git clone https://github.com/NVIDIA/apex}

\texttt{> cd apex}

\texttt{> pip install -v --disable-pip-version-check --no-cache-dir --global-option="--cpp\_ext" --global-option="--cuda\_ext" ./}
\\

\noindent 3. Download and preprocess data sets.

a) Cifar-10

\texttt{    > cd \$WORK/Ok-Topk/VGG/vgg\_data}

\texttt{    > wget https://www.cs.toronto.edu/$\sim$kriz/\\cifar-10-python.tar.gz}

\texttt{    > tar -zxvf cifar-10-python.tar.gz}

b) AN4

\texttt{    > cd \$WORK/Ok-Topk/LSTM/audio\_data}

\texttt{    > wget www.dropbox.com/s/l5w4up20u5pfjxf/an4.zip}

\texttt{    > unzip an4.zip}

c) Wikipedia

\texttt{    > cd \$WORK/Ok-Topk/BERT/bert/bert\_data}

\ \ Prepare the dataset according to the \texttt{README} file in this directory.

%%%%%%%%%%%%%%%%%%%%%%%%%%%%%%%%%%%%%%%%%%%%%%%%%%%%%%%%%%%%%%%%%%%%%
\subsection{Experiment workflow}

We run experiments on GPU clusters with \texttt{SLURM} job scheduler. In the job scripts of the artifact, we use \texttt{srun} to launch multiple processes among the computation nodes. But one can also use \texttt{mpirun} to launch multiple processes if there is no \texttt{SLURM} job scheduler on your machine. Once the jobs are finished, the results of training speed, test accuracy and training loss values for each algorithm will be output into \texttt{txt} files.

%%%%%%%%%%%%%%%%%%%%%%%%%%%%%%%%%%%%%%%%%%%%%%%%%%%%%%%%%%%%%%%%%%%%%
\subsection{Evaluation and expected result}

\noindent 1. To run VGG jobs

\texttt{> cd \$WORK/Ok-Topk/VGG}

\texttt{> ./sbatch\_vgg\_jobs.sh}
\\

\noindent 2. To run LSTM jobs

\texttt{> cd \$WORK/Ok-Topk/LSTM}

\texttt{> ./sbatch\_lstm\_jobs.sh}
\\

\noindent 3. To run BERT jobs

\texttt{> cd \$WORK/Ok-Topk/BERT/bert}

\texttt{> ./sbatch\_bert\_jobs.sh}
\\

\noindent 4. Check the output results

It takes about 5 hours for all jobs to be finished, which depends on the busyness of the job queue. To make sure all jobs have been finished, checking the status by:

\texttt{> squeue -u username}

If you see no job is running, then all jobs are finished.

\texttt{> cd \$WORK/Ok-Topk/VGG}

Check the 6 output \texttt{txt} files for the 6 algorithms (i.e., \texttt{dense, gaussiank, gtopk, oktopk, topkA,} and \texttt{topkDSA}), respectively, which contain the training speed and top-1 test accuracy for VGG.

\texttt{> cd \$WORK/Ok-Topk/LSTM}

Check the 6 output \texttt{txt} files for the 6 algorithms, respectively, which contain the training speed and WER values for LSTM.

\texttt{> cd \$WORK/Ok-Topk/BERT/bert}

Check the 6 output \texttt{txt} files for the 6 algorithms, respectively, which contain the training speed and training loss values for BERT.
\\

Users are expected to reproduce the results in this paper. Different software or hardware versions may lead to slightly variant results compared with the numbers reported in the paper, but it should not affect the general trends claimed in the paper, namely O\textit{k}-Top\textit{k} achieves higher training throughput and faster time-to-solution, and is more scalable than the other algorithms.

%%%%%%%%%%%%%%%%%%%%%%%%%%%%%%%%%%%%%%%%%%%%%%%%%%%%%%%%%%%%%%%%%%%%%
\subsection{Experiment customization}

Users can modify the variables in the job scripts to customize the experiments. For example, modify \textit{density} to change the density of the gradient; modify \textit{max-epochs} to change the number of training epochs; modify \textit{nworkers} and \textit{nodes} to change the number of processes used for training.

%%%%%%%%%%%%%%%%%%%%%%%%%%%%%%%%%%%%%%%%%%%%%%%%%%%%%%%%%%%%%%%%%%%%%
\subsection{Notes}
None. 

%%%%%%%%%%%%%%%%%%%%%%%%%%%%%%%%%%%%%%%%%%%%%%%%%%%%%%%%%%%%%%%%%%%%%
\subsection{Methodology}

Submission, reviewing and badging methodology:\\
\url{https://www.acm.org/publications/policies/artifactreview-badging}\\
\url{http://cTuning.org/ae/submission-20201122.html}\\
\url{http://cTuning.org/ae/reviewing-20201122.html}

\end{document}